\documentclass{article}
\usepackage[utf8]{inputenc}

\usepackage{enumitem}
\usepackage[colorlinks=true, pdfstartview=FitV, linkcolor=blue, citecolor=blue, urlcolor=blue]{hyperref}

\usepackage{amssymb,amsthm,latexsym}
\usepackage[left=25mm,right=25mm,top=30mm,bottom=30mm]{geometry}
\usepackage{amsmath,graphicx}
\usepackage[mathlines]{lineno}
\usepackage{float}
\usepackage{oubraces}
\usepackage{tabularx}
\usepackage{multirow}
\usepackage{tikz}
\usepackage{here}
\usepackage{cite}
\usepackage{todonotes}
\usepackage[capitalise]{cleveref}
\usepackage{setspace}
\usepackage{soul}
\usepackage{tikz}
\usetikzlibrary{calc}
\usetikzlibrary{decorations.pathmorphing}
\usetikzlibrary{patterns,shapes,positioning,matrix}
\makeatletter

\newcommand{\highlight@DoHighlight}{
  \fill [ decoration = {random steps, amplitude=1pt, segment length=15pt}
        , outer sep = -15pt, inner sep = 0pt, decorate
        , every highlighter, this highlighter ]
        ($(begin highlight)+(0,8pt)$) rectangle ($(end highlight)+(0,-3pt)$) ;
}

\newcommand{\highlight@BeginHighlight}{
  \coordinate (begin highlight) at (0,0) ;
}

\newcommand{\highlight@EndHighlight}{
  \coordinate (end highlight) at (0,0) ;
}

\newdimen\highlight@previous
\newdimen\highlight@current

\DeclareRobustCommand*\highlight[1][]{%
  \tikzset{this highlighter/.style={#1}}%
  \SOUL@setup
  \def\SOUL@preamble{%
    \begin{tikzpicture}[overlay, remember picture]
      \highlight@BeginHighlight
      \highlight@EndHighlight
    \end{tikzpicture}%
  }%
  \def\SOUL@postamble{%
    \begin{tikzpicture}[overlay, remember picture]
      \highlight@EndHighlight
      \highlight@DoHighlight
    \end{tikzpicture}%
  }%
  \def\SOUL@everyhyphen{%
    \discretionary{%
      \SOUL@setkern\SOUL@hyphkern
      \SOUL@sethyphenchar
      \tikz[overlay, remember picture] \emphlight@EndHighlight ;%
    }{%
    }{%
      \SOUL@setkern\SOUL@charkern
    }%
  }%
  \def\SOUL@everyexhyphen##1{%
    \SOUL@setkern\SOUL@hyphkern
    \hbox{##1}%
    \discretionary{%
      \tikz[overlay, remember picture] \highlight@EndHighlight ;%
    }{%
    }{%
      \SOUL@setkern\SOUL@charkern
    }%
  }%
  \def\SOUL@everysyllable{%
    \begin{tikzpicture}[overlay, remember picture]
      \path let \p0 = (begin highlight), \p1 = (0,0) in \pgfextra
        \global\highlight@previous=\y0
        \global\highlight@current =\y1
      \endpgfextra (0,0) ;
      \ifdim\highlight@current < \highlight@previous
        \highlight@DoHighlight
        \highlight@BeginHighlight
      \fi
    \end{tikzpicture}%
    \the\SOUL@syllable
    \tikz[overlay, remember picture] \highlight@EndHighlight ;%
  }%
  \SOUL@
}
\makeatother

\usepackage[most]{tcolorbox}
\tcbuselibrary{listings,theorems}

\usepackage{mdframed}

\numberwithin{equation}{section}

\theoremstyle{plain}
\newtheorem{theorem}{Theorem}

\newtheorem{lemma}[theorem]{Lemma}
\newtheorem{proposition}[theorem]{Proposition}
\newtheorem{definition}{Definition}
\newtheorem{example}[theorem]{Example}

\newtheorem{conjecture}[theorem]{Conjecture}
\newtheorem{problem}[theorem]{Problem}

\theoremstyle{definition}

\theoremstyle{remark}
\newtheorem{remark}[theorem]{Remark}

\newcommand{\Z}{{\mathbb Z}}

\renewcommand{\epsilon}{\varepsilon}
\renewcommand{\phi}{\varphi}

\DeclareMathOperator{\Suff}{Suff}
\DeclareMathOperator{\Pref}{Pref}
\DeclareMathOperator{\Fact}{Fact}
\renewcommand{\alph}{\textit{alph}}

\DeclareMathOperator{\Pal}{Pal}
\DeclareMathOperator{\AS}{AS}
\def\PV{\mathcal{P}}

\def\tr{\textit{tr}}
\def\tm{\textit{tm}}
\def\vtm{\textit{vtm}}

\def\cd3#1{\textbf{\textsf{#1}}}

\newcommand{\ass}[2]{\AS_{#2}(#1)}
\newcommand{\as}[1]{\AS(#1)}

\def\act{\chi_{ab}}

\title{Abelian Combinatorics on Words: a Survey}
\author{Gabriele Fici \and Svetlana Puzynina\thanks{The second author has been supported by Russian Foundation of Basic Research
(grant 20-01-00488).}}
\date{\today}

\begin{document}

\maketitle

\begin{abstract}
We survey known results and open problems in abelian combinatorics on words. Abelian combinatorics on words is the extension to the commutative setting of the classical theory of combinatorics on words. The extension is based on \emph{abelian equivalence}, which is the equivalence relation defined in the set of words by having the same Parikh vector, that is, the same number of occurrences of each letter of the alphabet. In the past few years, there was a lot of research on abelian analogues of classical definitions and properties in combinatorics on words. This survey aims to gather these results.  
\end{abstract}

\tableofcontents

\section{Introduction}

Combinatorics on words is the algebraic study of symbolic sequences. Given a finite set $\Sigma$, called the alphabet, one can construct a free monoid $\Sigma^*$ by equipping the alphabet with the operation of concatenation. The specificity of the theory is that the concatenation operation is not commutative. However, over the past years researchers built up a commutative theory by restricting the attention to the number of occurrences of letters in a word rather than to the order in which they appear. Formally speaking, fixed an ordered alphabet $\Sigma_d=\{0,1,\ldots,d-1\}$, of cardinality $d>1$\footnote{The case of a unary alphabet, $d=1$, is not interesting in this context, since it can be trivially reduced to elementary arithmetic.}, one can consider the map $\PV$ defined on the set $\Sigma_d^*$ of words over $\Sigma_d$ by $\PV(w)=(|w|_0,|w|_{1},\ldots,|w|_{d-1})$, where $|w|_i$ denotes the number of occurrences of $i$ in the word $w$. The equivalence relation induced by $\PV$ is called the \emph{abelian equivalence} and is at the basis of the theory of abelian combinatorics on words. A natural aim of the theory is to extend to the abelian setting the main definitions and results that have been introduced and discovered in the noncommutative setting. 
In the past few years, there was a lot of research on abelian analogues of classical definitions and properties in combinatorics on words. This survey aims to gather these results.

Some specific aspects of abelian combinatorics on words have already been the object of a survey, e.g., avoidance. In these cases, we shortly recall the results. Otherwise, we present a more detailed description of the results and give references to the published papers. Even if our aim is to give a comprehensive view of the subject, in some cases we decided to select only the results that we find more relevant or interesting. 

The content of the document covers the following topics: Among the various definitions of complexity for an infinite word, the classical one is the factor complexity, i.e., the integer function that counts for each $n$ the number of distinct factors of length $n$ occurring in the word. When this number of factors is counted up to abelian equivalence, one has the notion of abelian complexity, which is the object of Sec.~\ref{sec:complexity}. The study of repetitions in words has been generalized to the abelian setting and is the object of Sec.~\ref{sec:repetitions}.
Another classical subject in combinatorics of words is that of avoidability of patterns. The extension of avoidability to the abelian setting is the object of Sec.~\ref{sec:avoidability}. For finite words, many of the classical results about periods and borders have an abelian counterpart, presented in Sec.~\ref{sec:borders}. One of the most studied classes of words is that of Sturmian words; Abelian properties of Sturmian words are the object of Sec.~\ref{sec:sturmian}. In Sec.~\ref{sec:modifications}, we present some interesting modifications of the abelian equivalence, e.g., $k$-abelian equivalence, weak abelian equivalence and $k$-binomial equivalence. Finally, in Sec.~\ref{sec:misc}, we present some  results on abelian counterparts of specific definitions, e.g., subshifts or palindromic richness. 

Combinatorics on words has found applications in several disciplines, e.g., text processing, bioinformatics, error-correction codes, etc. The field referring to the applications of combinatorial properties of words to the design of efficient algorithms for string pattern matching is sometimes referred to as ``stringology''. Recently, some of the results in abelian combinatorics on words we present in this survey have found applications in what is called ``abelian stringology'', or abelian pattern matching (also called jumbled pattern matching). We do not cover these results, however, since the focus of this survey is more on the algebraic aspects of the theory. 

In this document, we briefly recall the basic definitions of the classical theory of combinatorics on words, pointing the reader to the classical books on the subject for further details. In this way, the document remains self-contained and accessible to all interested readers. We made a particular effort in trying to make the notation uniform all along the presentation of the results. For this reason, the notation we use in this document may differ with respect to that used in the original papers.

\section{Preliminaries}\label{sec:prelim}

We start by recalling some standard definitions. For other basics of combinatorics on words 
we refer the reader to the classical books on the subject~\cite{Lot01,LothaireAlg,Lot05,Pitheasfogg,AS03,DBLP:reference/hfl/ChoffrutK97,berthe_rigo_CANT,Rigo2014}.

Let $f,g$ be integer functions. We write $f\in O(g)$ if there exists a constant $C>0$ such that for every $n$, $f(n)\leq Cg(n)$. In this case we also write $g\in \Omega(f)$. If $f\in O(g)$ and $g\in O(f)$, we write $f\in \Theta(g)$. If $\lim_{n\to \infty}f(n)/g(n)=0$, we write $f\in o(g)$. 


We let  $\Sigma_d=\{0,1,\ldots, d-1\}$ denote a $d$-ary alphabet. A \emph{word} over $\Sigma_d$ is a concatenation of letters from $\Sigma_d$. The length of a word $w$ is denoted by $|w|$. The empty word $\varepsilon$ has length $0$. The set of all words (resp. all nonempty words) over  $\Sigma_d$ is denoted  $\Sigma_d^*$ (resp.~by $\Sigma_d^+$), while the set of all words of length equal to $n$ over  $\Sigma_d$ is denoted by $\Sigma_d^n$. 

Let $w=uv$, with $u,v\in \Sigma_d^*$. We say that $u$ is a \emph{prefix} of $w$ and that $v$ is a \emph{suffix} of $w$. A \emph{factor} of $w$ is a prefix of a suffix (or, equivalently, a suffix of a prefix) of $w$. The sets of prefixes, suffixes, factors of a word $w$ are denoted, respectively, by $\Pref(w),\Suff(w),\Fact(w)$. We also use  $\Pref_k(w),\Suff_k(w)$ to denote the prefix and the suffix of length $k$ of $w$, respectively. 


An integer $p$ is \emph{a period} of a word $w=w_1w_2\cdots w_n$, $w_i\in \Sigma_d$, if $w_i=w_j$ whenever $i=j \mod p$. We call \emph{the period} of $w$ the smallest of its periods,  denoted  $\pi(w)$. The \emph{exponent} of a word $w$ is the ratio $|w|/\pi(w)$. 

A word $w$ is a \emph{$k$-power}, $k>1$, if it has length $kp$ and $p>0$ is a period of $w$. A $2$-power is simply called 
a \emph{square} and a $3$-power is simply called a \emph{cube}. A word that is not a $k$-power for any $k>1$ is called \emph{primitive}.

Let $w$ be a $k$-power. Then its prefix $u$ of length $|w|/k$ is called \emph{a root} of $w$. The \emph{primitive root} of the word $w$ is the shortest of its roots, that is, its prefix of length $\pi(w)$ (notice that the primitive root of a word is a primitive word). For example, the word $aaaa$, $a\in\Sigma_d$, is both a square (with root $aa$) and a $4$-power, with primitive root $a$.

We say that a nonempty word $v$ is a \emph{border} of a word $w\neq v$ if $w=vu=u'v$ for some words $u,u'$. It follows from the definition that $v$ is a border of $w$ if and only if $|w|-|v|$ is a period of $w$ of length $\leq |w|$.

A nonempty word $w=w_1w_2\cdots w_n$, $w_i\in \Sigma_d$, is a \emph{palindrome} if  it coincides with its \emph{reversal} $w^R=w_nw_{n-1}\cdots w_1$. The empty word is also assumed to be a palindrome. The set $\Pal(w)$ of palindromic factors of $w$ has cardinality $|\Pal(w)|\leq |w|+1$;  $w$ is called \emph{rich} if the equality holds.


An \emph{infinite word} (or \emph{right-infinite word}) over $\Sigma_d$ is a non-ending concatenation of letters from $\Sigma_d$.
An infinite word is called \emph{purely periodic} if it has a period, i.e., it can be written as $v^\omega$ for some finite word $v$ (the notation $u^{\omega}$ stands for $uuu \cdots $); \emph{ultimately periodic} if it has a purely periodic infinite suffix, i.e., it can be written as $uv^\omega$ for some finite words $u,v$; or \emph{aperiodic} otherwise, i.e., if it is not ultimately periodic. 

An infinite word $x$ is \emph{recurrent} if every finite factor of $x$ occurs in $x$ infinitely often;  \emph{uniformly recurrent} if for every finite factor $u$ of $x$ there exists an integer $N$ (that depends on $u$) such that $u$ occurs in every factor of $x$ of length $N$;  \emph{linearly recurrent} if there exists an integer $m$ such that for every finite factor $u$ of $x$, $u$ occurs in every factor of $x$ of length $m|u|$.

A \emph{substitution} is a map $h$ from $\Sigma_d$ to $\Sigma_d^*$ such that the image of every letter is nonempty. The notion of a substitution is generalized from letters to words in a natural way by concatenation: $h(uv)=h(u)h(v)$. If for a letter $a\in \Sigma_d$,  $h(a)$ is a word of length at least $2$ beginning with $a$, the substitution has a unique fixed point beginning with $a$, which is the infinite word $\lim_{n\to\infty}h^n(a)$.
An infinite word is called \emph{purely morphic} if it is a fixed point of a substitution.

A substitution is \emph{uniform} if all the images have the same length and \emph{primitive} if for every letter $a$ there exists an iterate of the substitution on $a$ that contains all the letters of $\Sigma_d$. Fixed points of primitive substitutions are known to be linearly recurrent. 

More generally, given two alphabets $\Sigma$ and $\Delta$, a \emph{morphism} between  $\Sigma^*$ and $\Delta^*$ is a map $h$ such that for every $u,v\in\Sigma^*$, one has $h(uv)=h(u)h(v)$. A morphism can be specified by giving the list of images of letters in $\Sigma$. A morphism is \emph{non-erasing} if the images of all letters are nonempty. Notice that a substitution is therefore a non-erasing endomorphism.  

Given an infinite  word $x$ over $\Sigma_d$, the \emph{factor complexity} of $x$ is the integer function $p_{x}(n)=|\Fact(x)\cap \Sigma_d^n|$ counting the number of distinct factors of length $n$ of $x$, for each $n\geq 0$.

An infinite word is aperiodic if and only if its factor complexity is unbounded. In particular, a classical result of Morse and Hedlund~\cite{MoHe38} is that the factor complexity of an aperiodic word $x$ verifies $p_{x}(n)\geq n+1$ for every $n$. An aperiodic word with minimal factor complexity $p_{x}(n)= n+1$ for every $n$ is called a \emph{Sturmian word}. A famous example of Sturmian word is the \emph{Fibonacci word} $$f=010010100100101001\cdots$$ which can be obtained as the fixed point of the substitution $0\mapsto 01, 1 \mapsto 0$. 
Sturmian words can be defined in many equivalent ways; in particular, via balance, iterated palindromic closure and Sturmian morphisms. 

The most natural generalization of Sturmian words to nonbinary alphabets, which shares many structural properties of Sturmian words, is \emph{Arnoux--Rauzy words}, or strict episturmian words~\cite{ArRa,DJP}. One of the ways to define Arnoux--Rauzy words  --- and in particular Sturmian words --- is via iterated palindromic closure. 
The \emph{right palindromic closure} of a finite word $u\in\Sigma_d^*$, denoted
by $u^{(+)}$, is the shortest palindrome that has $u$ as a prefix.
The \emph{iterated (right) palindromic closure operator} $\psi$ is
 defined recursively by the following rules:
\[ \psi(\varepsilon)=\varepsilon, \quad \psi(ua)=(\psi(vu)a)^{(+)}\]
for all $u \in  \Sigma_d^*$ and $a \in  \Sigma_d$. For example, $\psi(0110)=0101001010$. 

The definition of $\psi$ can
be extended to infinite words over $ \Sigma_d$, $d\geq 2$, as follows:
$\psi(u)=\lim_{n\to\infty} \psi(\mbox{Pref}_n (u))$, i.e., $\psi(u)$ is the
infinite word having $\psi(\mbox{Pref}_n (u))$ as its prefix for
every $n \in \mathbb N$.
Let $u$ be an infinite word over the alphabet $\Sigma_d$ such
that every letter occurs infinitely often in $u$. The word
$x=\psi(u)$ is then called a \emph{characteristic (or
standard) Arnoux--Rauzy word} and $u$ is called the
\emph{directive sequence} of $x$. An infinite word $x$ is called
an Arnoux--Rauzy word if it has the same set of factors of a
(unique) characteristic Arnoux--Rauzy word. For $d=2$, this gives an equivalent definition of Sturmian words~\cite{deLuca1996:sturmian_words_structure_combinatorics_arithmetics}. 
An example of characteristic Arnoux--Rauzy word is given by the \emph{Tribonacci word} \[\tr=010201001020101020100\cdots\] which has directive sequence $(012)^\omega$. The Tribonacci word can also be defined as the fixed point of the substitution $0 \mapsto 01, 1\mapsto 02, 2 \mapsto 0$.

The \emph{critical exponent} $\chi(x)$ of an infinite word $x$ is the supremum of the exponents of its factors. We say that an infinite word $x$ is $\beta$-free (resp.~$\beta^+$-free), for a real number $\beta$, if no factor has exponent $\beta$ or larger (resp.,~if no factor has exponent larger than $\beta$). For example, the critical exponent of the Fibonacci word is $2+\phi$, where $\phi=(1+\sqrt{5})/2$ is the golden ratio~\cite{MignosiPirillo}; hence the Fibonacci word is $(2+\phi)$-free (and in particular $4$-free).

It is a trivial fact that every word over $\Sigma_2$ of length at least $4$ contains a square factor, so there do not exist infinite binary square-free words. Still, there exist infinite binary words that are $2^+$-free. An example is the \emph{Thue--Morse word}
\[\tm=01101001100101101001\cdots\]
which can be obtained as the fixed point starting with $0$ of the substitution $0\mapsto 01, 1 \mapsto 10$.

Another famous word we will mention in this paper is the \emph{regular paperfolding word}:
$$p=001001100011011000100\cdots$$ 
which, contrarily to the Fibonacci and the Thue--Morse words, cannot be obtained as the fixed point of a substitution. It can be defined as a \emph{Toeplitz word} with pattern $v = 0?1?$, that is, starting from the word $v^\omega$, we replace the occurrences of the characters $?$ with the word $v^\omega$, then in the new word we again replace the remaining occurrences of $?$ with $v^\omega$ and so on, thus defining a word without $?$. 

More generally, one can construct an infinite (actually, uncountable) family of words, called \emph{paperfolding words}, by alternating the replacements of the occurrences of $?$ with $v_0^\omega$ and $v_1^\omega$, where $v_0=0?1?$ and $v_1=1?0?$, according to a binary directive sequence.


We will need a symbolic dynamical notion of the subshift generated by an infinite word. A \emph{subshift} $X\subseteq \Sigma_d^{\mathbb{N}}$, $X\neq \emptyset$, is
a closed set (with respect to the product topology of
$\Sigma_d^{\mathbb{N}}$) and is invariant under the shift operator $\sigma$, defined by $\sigma(a_0a_1a_2\cdots) = a_1a_2\cdots$, that is,
$\sigma(X) \subseteq X$. We call $\Sigma_d^{\mathbb{N}}$ the \emph{full
shift} over $\Sigma_d$. A subshift $X\subseteq \Sigma_d^{\mathbb{N}}$ is
called \emph{minimal} if $X$ does not contain any proper
subshifts. For a subshift $X\subseteq \Sigma_d^{\mathbb{N}}$ we let 
$\Fact(X) = \cup_{ y\in X}\Fact( y)$.
Let $ x \in \Sigma_d^{\mathbb{N}}$. We let $\Omega_{x}$ denote the
\emph{shift orbit closure} of $x$, that is, the set $\{ y\in
\Sigma_d^{\mathbb{N}} \colon \Fact( y)\subseteq \Fact( x)\}$. Thus,
$\Fact(\Omega_{ x}) = \Fact( x)$ for an infinite word $
x\in\Sigma_d^{\mathbb{N}}$. It is known that $\Omega_{ x}$ is minimal if and
only if $ x$ is uniformly recurrent. See \cite{LindMarcus95} for more on
the topic.

Given a  word $w$ over  $\Sigma_d$, we let $|w|_{i}$ denote
the number of occurrences of the letter $i$ of  $\Sigma_d$ in $w$.
The \emph{Parikh vector} (also called \emph{composition vector} or  \emph{abelianization})
of the word $w$ is the vector
$\PV(w)=(|w|_0,|w|_{1},\ldots,|w|_{d-1})$, counting the
occurrences of the letters of $\Sigma_d$  in $w$. 

Two words have the same Parikh vector if and only if one is an anagram of the other. In particular, if two words have the same Parikh vector, then they must have the same length, which is also the sum of the components of the Parikh vector (called the \emph{norm} of the
Parikh vector).

\begin{definition}
The equivalence relation $\sim_{ab}$ defined on $\Sigma_d^*$ by the property of having the same Parikh vector is called
\emph{abelian equivalence}.
\end{definition}

 For example, the words $01101$ and $10011$ are abelian equivalent, while the words $01101$ and $10010$ are not.  Besides combinatorics on words, the concepts of Parikh vector (and Parikh matrix) and  abelian equivalence are used in semigroup theory and are applied in formal language theory; see, e.g., Parikh theorem for context-free languages \cite{10.1145/321356.321364}.

\section{Abelian complexity}\label{sec:complexity}

In this section,  we discuss abelian modifications of the classical notion of factor complexity of an infinite word and of the pattern complexity introduced by Kamae and Zamboni~\cite{kamae_zamboni_2002}.

\subsection{Abelian complexity and periodicity}

The \emph{abelian complexity} of the word $x$ over $\Sigma_d$ is the integer function
$$a_{x}(n)=\left | ({\Fact(x) \cap \Sigma_d^{n}}) / {\sim_{ab}}
\right |,$$ where $\sim_{ab}$ is the abelian equivalence, i.e.,
$a_{x}$ is the function that counts the number of distinct Parikh
vectors of factors of length $n$ of $x$, for every $n\geq 0$.

If an infinite word $x$ is ultimately periodic, then its
abelian complexity is bounded. Indeed, by the  Morse--Hedlund theorem, the usual factor complexity of
ultimately periodic words is bounded, and the abelian complexity
cannot be greater than the factor complexity, since the identity is a refinement of the abelian equivalence.

On the other hand, there exist aperiodic words with bounded
abelian complexity. As an example, all Sturmian words are aperiodic
and have abelian complexity equal to $2$, as we will see in Section~\ref{sec:sturmian}. In fact, it is easy to see that an aperiodic word cannot have an abelian complexity equal to $1$ for any $n$: 

\begin{lemma}
 If there exists $n>0$ such that  $a_{x}(n)=1$, then $x$ is purely periodic. More precisely, the smallest period of $x$ is the least such $n$.
\end{lemma}

\begin{proof}
 Let $n$ be the least integer such that $a_{x}(n)=1$, that is, all the factors of $x$ of length $n$ have the same Parikh vector. In particular, the prefix $x_1x_2\cdots x_n$ of length $n$ of $x$ has the same Parikh vector as the factor $x_2x_3\cdots x_{n+1}$. This implies that $x_{n+1}=x_1$. Analogously, one deduces that $x_{n+2}=x_2$ and so on. We therefore have that $x$ has period $n$.
\end{proof}

The maximal abelian complexity is realized, for example, by words with full factor complexity, like, e.g., the \emph{binary Champernowne word} $0\,1\,10\,11\,100\,101\,110\,111\cdots$ obtained by concatenating the binary expansions of the natural numbers in the natural order (with zero represented by $0$). We have:

\begin{theorem}
 For all infinite words $x$ over $\Sigma_d$, and for all $n \ge 0$,
\[1\le a_x(n)\le \binom{n+d-1}{d-1}.\]
In particular, the abelian complexity is bounded by $O(n^d)$.
\end{theorem}

\begin{proof}
The maximum value of the abelian complexity of a word over $\Sigma_d$ is the maximum
number of ways of writing $n$ as the sum of $d$ nonnegative integers. This well-known number is called the number of compositions of $n$ into $d$ parts and its value is given
by the binomial coefficient $\binom{n+d-1}{d-1}$.
\end{proof}

However, there exist  infinite binary words having maximal abelian complexity but linear factor
complexity. For example, take the alphabet $\Delta=\{a,b,c\}$ and let $f$ and $g$ be the morphisms defined by $f(a) = abc$,
$f(b) = bbb$, $f(c) = ccc$, $g(a) = 0 = g(c)$ and $g(b) = 1$. Let $x$ be the
fixed point of $f$ beginning in $a$. Then the image of $x$ under $g$ is the word
\[x'=0\prod_{i\ge 0}1^{3^i}0^{3^i}\]
The word $x'$ has maximal abelian complexity but linear factor complexity.

\begin{definition}\label{def:balance}
 A (finite or infinite) word $w$ over $\Sigma_d$ is $C$-balanced for an integer $C>0$ if for every letter $a\in\Sigma_d$ and every two factors $u,v$ of $w$ of the same length, one has $||u|_a-|v|_a|\leq C$.
For $C=1$, the constant is usually omitted and the word is simply called
balanced.

The balance function of $w$ is the function
\[B_w(n)=\max_{a\in \Sigma_d}\ \ \max_{u,v\in \Fact(w)\cap\Sigma_d^n}||u|_a-|v|_a|.\]
Clearly, a word is $C$-balanced if and only if its balance function is bounded by $C$.
\end{definition}

In other words, a word is $C$-balanced if for every letter $a$, taking a window of any fixed size sliding on the word one has a number of $a$'s falling in the window that ranges from a minimal value $k$, depending on the size of the window, to a maximal value $k+C$. An immediate consequence of this remark is the following:

\begin{proposition}
 Let $x$ be an infinite word. Then the abelian complexity of $x$ is bounded if and only if $x$ is $C$-balanced for some $C>0$.
\end{proposition}


A well-known result by Coven and Hedlund~\cite{Coven-Hedlund} states that a binary aperiodic word is 1-balanced if and only if it is Sturmian, which can be reformulated in terms of abelian complexity as follows:

\begin{theorem}
Let $x$ be an aperiodic binary  word. Then $x$ is Sturmian 
if and only if $a_x(n)=2$ for every $n\ge 1$. 
\end{theorem}

\subsection{Abelian complexity of some families of words}

We start with the Thue--Morse word $\tm$. Its abelian complexity is given by:
 \[a_{tm}(n) = \left\{ \begin{array}{lllll}
2 & \mbox{if $n$ is odd,}\\
3 & \mbox{if $n$ is even.}\\
\end{array} \right.\]
Indeed, the Thue--Morse word consists of blocks 01 and 10, so its factors of odd length contain several blocks plus either 1 or 0, hence abelian complexity is 2. For even lengths, a factor  contains either several complete blocks, or several complete blocks plus two letters, which can be both 0, both 1 or 0 and 1, the latter case giving the same abelian class as factors consisting of full blocks; hence the abelian complexity is 3 for even lengths. 

The abelian complexity together with the factor complexity almost characterize the  Thue--Morse word, in the sense that
an infinite word has the same abelian and factor complexity as the Thue--Morse word
if and only if it is in its shift orbit closure \cite{Richomme201179}.


Let $\tr$ be the Tribonacci word. 
For every $n\ge 1$, $a_{\tr}(n)\in\{3,4,5,6,7\}$. Moreover,  each of these five values is assumed\cite{DBLP:journals/aam/RichommeSZ10}, and the exact value of $a_{\tr}(n)$ can be effectively computed~\cite{Turek15,Shallit21}. However, in general, Arnoux--Rauzy words can have unbounded abelian complexity (or, equivalently, there exist Arnoux--Rauzy words which are not $C$-balanced for any $C$)~\cite{AIF_2000__50_4_1265_0}.
 
 Madill and Rampersad~\cite{DBLP:journals/dm/MadillR13} studied the abelian complexity of the regular paperfolding word $p$ and  characterized it by proving the following recursive relations:
 \begin{eqnarray*}
 & & a_p(4n)  =  a_p(2n) \\
 & & a_p(4n+2) = a_p(2n+1)+1 \\
  & & a_p(16n+1) =  a_p(8n+1) \\
 & & a_p(16n+\{3,7,9,13\}) =  a_p(2n+1)+2 \\
 & & a_p(16n+5) =  a_p(4n+1)+2 \\
 & & a_p(16n+11) =  a_p(4n+3)+2 \\
 & & a_p(16n+15) =  a_p(2n+2)+1.
\end{eqnarray*}
From these formulas, it follows that the regular paperfolding word has unbounded abelian complexity.

Blanchet-Sadri et al.~studied the abelian complexity of the ternary squarefree word of Thue (also called \emph{Hall word}, or \emph{Variant of Thue--Morse}) $\vtm=012021012102012\cdots$ --- which can be obtained as the fixed point of the substitution $0 \mapsto 012, 1\mapsto 02, 2 \mapsto 1$ --- and that of the \emph{period-doubling word} $pd=01000101010001000\cdots$, which is equal to $\vtm$ modulo $2$ and is the fixed point of the substitution  $0 \mapsto 01, 1\mapsto 00$~\cite{DBLP:journals/int/Blanchet-SadriC14}.

Rauzy\cite{Ra82} asked whether an infinite word exists with constant abelian complexity equal to~$3$. Richomme, Saari and Zamboni~\cite{Richomme201179} answered  this question positively by showing that any aperiodic balanced word over $\Sigma_3$ has this property. It has been proved that there are no recurrent words of
constant abelian complexity~4~\cite{DBLP:journals/aam/CurrieR11}. However, for every integer $c\geq 2$, there is a  recurrent
word $x$ with abelian complexity $a_x(n)=c$  for every $n\geq c-1$. 
\cite{DBLP:journals/jalc/Saarela09}.

We  now discuss the (abelian) complexity of purely morphic words. 
A well-known classification of factor complexities of
fixed points of morphisms has been obtained in a series of papers, finally completed by Pansiot \cite{Pansiot}, and states that
there are 5 classes of possible complexity growths: $\Theta(1)$,
$\Theta(n)$, $\Theta(n\log n)$, $\Theta(n\log\log n)$ and
$\Theta(n^2)$.
The abelian complexity of purely morphic words is more complicated
and is completely classified only for fixed points of binary
morphisms (more precisely, only the superior limit of the abelian complexity has been classified). 

The balance function of primitive morphic words has
been characterized by Adamczewski \cite{DBLP:journals/tcs/Adamczewski03}. As an
immediate corollary of this characterization, we get a
classification of  abelian complexities of fixed points of primitive binary
morphisms. For integer functions $f$ and $g$, let us write $f(n)=\Omega'(g(n))$ if
$\limsup_{n\to\infty} f(n)/g(n)>0$. Then the abelian complexity of a purely
morphic word is either $\Theta(1)$, or $(O\cap\Omega')(\log n)$, or
$(O\cap\Omega')(n\log_{\theta_1} \theta_2)$, where $\theta_1$ and
$\theta_2$ are the first and second largest eigenvalues
of the adjacency matrix of the morphism. \footnote{We cannot write $\Theta$ in place of $(O\cap\Omega')$ because the functions could be oscillating; however, here we are essentially interested in their maximum values.} 

A classification of abelian complexities of fixed points of
non-primitive binary morphisms is due to Blanchet-Sadri, Fox and
Rampersad \cite{DBLP:journals/aam/Blanchet-SadriF14} and completed by Whiteland \cite{DBLP:journals/jalc/Whiteland19}: this
can be either $\Theta(1)$, or $\Theta(n)$, or $\Theta(n/\log n)$,
or $\Theta(n^{\log_k l})$ with $1<k<l$, or it can fluctuate
between $\Theta(1)$ and $\Theta(\log(n))$. Some algorithmic aspects of computing  the abelian complexity of fixed points of uniform morphisms have been studied in~\cite{DBLP:journals/tcs/Blanchet-SadriS16}.




 \subsection{Abelian pattern complexity}

The \emph{pattern complexity}, a modification of the notion of factor complexity, introduced by Kamae and Zamboni~\cite{kamae_zamboni_2002}, can also be well generalized to the abelian setting. A \emph{pattern} $S$ is a $k$-element subset of nonnegative integers: $S=\{s_1<s_2< \dots < s_k\}$. For an infinite word $w$, we put 

$$w[S]=w_{s_1} w_{s_2} \cdots w_{s_k}.$$

For each $n$, the word $w[n+S]$ is called an $S$-factor of $w$, where $n+S = \{n+s_1, n+s_2, \dots, n+
s_k\}$. We  let $F_w(S)$ denote the set of all $S$-factors of $w$. The pattern complexity $patt_w(S)$ is then defined by 
\[patt_w(S) = |F_w(S)|,\] and the \emph{maximal pattern complexity} $patt^*_w(k)$ by 
\[patt^*_w(k)= \sup_{\substack{S\subset \mathbb{N}\\|S|=k}} patt_w(S).\]

An infinite word $w$ over $\Sigma_d$ is called \emph{periodic by projection} if there exists a nonempty set $B \subsetneq \Sigma_d$  such that
$1_B(w) = 1_B(w_0)1_B(w_1)1_B(w_2) \cdots \in \{0, 1\}^{\mathbb{N}}$
is ultimately periodic (where $1_B$ denotes the characteristic function of $B$). A word is \emph{aperiodic
by projection} if it is not periodic by projection. The following connection between the maximal pattern complexity and periodicity is known:

\begin{theorem}{\cite{DBLP:journals/ejc/KamaeH06}} \label{th:maxpatt} Let $w$ be an infinite aperiodic by projection word over $\Sigma_d$, $d \geq 2$. Then for every
positive integer $k$, $patt^*_w(k) \geq dk$.\end{theorem}

We can then define an abelian analogue of the notion of the pattern complexity by 
\[patt^{ab}_w(S) = |F_w(S)/\sim_{ab}|,\] 
and the \emph{maximal pattern abelian complexity} $patt^{*ab}_w(k)$ by \[patt^{*ab}_w(k)= \sup_{S\subset \mathbb{N}, |S|=k} patt^{ab}_w(S).\]

Then the following abelian analogue of Theorem \ref{th:maxpatt} holds:

\begin{theorem} \cite{kamae_widmer_zamboni_2015} \label{th:maxpatt_ab} Let  $w$ be a recurrent and aperiodic by projection infinite word over $\Sigma_d$, $d\geq 2$. Then for
every positive integer $k$, 
$$patt^{*ab}_w(k) \geq (d-1)k + 1.$$
When $d = 2$, the equality always holds. Moreover, for $k = 2$ and general $d$, there exists $w$ satisfying
the equality.\end{theorem}

In the abelian case, the condition of recurrence is necessary, since there exist non-recurrent counterexamples satisfying the inequality.

\section{Abelian repetitions}\label{sec:repetitions}

Recall that an abelian square is a nonempty word of the form $uv$, where $u$
and $v$ are abelian equivalent, i.e.,  have the same Parikh
vector. For example, $0110110011$ is an abelian square: $01101\sim_{ab} 10011$. More generally, an abelian $k$-power is a word of the form
$u_1u_2\cdots u_k$, where all the $u_i$ have
the same Parikh vector. An asymptotic estimate of the number of abelian squares of length $n$ has been given in \cite{DBLP:journals/combinatorics/RichmondS09}.

\subsection{Abelian complexity and abelian powers}

There is a relationship between abelian complexity and abelian powers, stated in the following theorem:

\begin{theorem}\cite{Richomme201179}\label{boundedAC}
 If a word has bounded abelian complexity, then it contains abelian $k$-powers for every $k>1$.
\end{theorem}

However, this is not a characterization of words with bounded abelian complexity. 
Indeed, Holub proved that all paperfolding words contain abelian powers of every order, and paperfolding words have unbounded abelian complexity. 

\begin{theorem}\cite{DBLP:journals/jct/Holub13}\label{paperfoldingpowers}
 All paperfolding words contain abelian $k$-powers for every $k>1$.
\end{theorem}

 
 In the case of the Thue--Morse word, we even have that every infinite suffix begins with an abelian
$k$-power for every positive integer $k$. However, it is possible to construct a uniformly recurrent binary word with bounded abelian complexity such that none of its prefixes is an abelian square~\cite{cassaigne2011avoiding}. 

\subsection{Abelian critical exponent}

Recall that the critical exponent $\chi(x)$ of an infinite word $x$ is the supremum of rational numbers $\beta$ such that $u^{\beta}$ occurs in $x$ for some factor $u$ of $x$. 
Notice that the critical exponent of an infinite word can be infinite, as in the case, for example,  for any (ultimately) periodic word. 

The following theorem was proved by Krieger and Shallit~\cite{DBLP:journals/tcs/KriegerS07}.

\begin{theorem}
 The following statements hold:
\begin{enumerate}
 \item For every real number $\beta>1$  there exists an infinite word over some alphabet whose critical exponent is $\beta$; 
 \item For every real number $\beta\ge 2$  there exists an infinite binary word whose critical exponent is $\beta$. 
\end{enumerate}
\end{theorem}

The maximum exponent of an abelian power occurring in
an infinite word does not give any interesting information on abelian powers, e.g.,  in words with bounded abelian complexity. Therefore, the following generalization to the abelian case has been proposed~\cite{tcs16}:

\begin{definition}\label{def:ace}
  Let $x$ be an infinite word. For every integer $m>1$, let $k_{m}$  be the
  maximum exponent of an abelian power  of period $m$ in $x$. The \emph{abelian
  critical exponent of $x$} is defined as
  \begin{equation}\label{eq:1}
    \act(x) = \limsup_{m \to \infty} \frac{k_{m}}{m}.
  \end{equation}
\end{definition}

Peltom\"aki and Whiteland proved the following result:

\begin{theorem}\cite{DBLP:conf/mfcs/PeltomakiW20} 
For every nonnegative real number $\beta$ there
exists an infinite binary word having abelian critical exponent $\beta$. 
\end{theorem}

We will see in a later section that for every nonnegative real number $\beta$ greater than a constant $c_F\simeq 4.53$ there exists a Sturmian word having abelian critical exponent $\beta$.


\subsection{Abelian square factors}\label{subsec:squares}

  In this section, we consider the problem of counting the number of abelian squares in a word of length $n$. For classical squares, their maximal number in a word of length $n$ is less than $n$. More precisely,  Fraenkel and Simpson~\cite{DBLP:journals/jct/FraenkelS98} showed  that a word of length $n$  contains less than $2n$  distinct squares, and conjectured that the bound is actually $n$. After several improvements, the conjecture was recently solved by Brlek and Li~\cite{BL22} (see also~\cite{DBLP:journals/aam/Li22}). 

As for the number of abelian square factors, it is easy to see that a word of length $n$ can contain $\Theta(n^2)$ distinct abelian square factors; e.g., words of the form $0^m10^m10^m$. 


 If one considers only abelian squares that are not abelian equivalent, then it can be shown that a word of length $n$ can contain $\Theta(n^{3/2})$ nonequivalent abelian square factors~\cite{DBLP:journals/tcs/KociumakaRRW16}. It is conjectured that a word of length $n$ always contains $O(n^{3/2})$ nonequivalent abelian square factors. 

For other open problems on abelian squares the reader is referenced to~\cite{simpson2018solved}.

The largest number of distinct abelian square factors in an infinite word has also been studied.
We need some notation. Given a finite or infinite word $w$, we
let $\ass{w} {n}$ denote the number of distinct abelian-square factors of
$w$ of length $n$. Of course,  $\ass{w} {n} =0$ if $n$ is odd, so this
quantity is significant only for even values of $n$. Furthermore, for a
finite word $w$ of length $n$, we let $\as{w} =\sum_{m\leq n} \ass{w} {m}$ denote the
total number of distinct abelian-square factors, of all lengths, in $w$.

\begin{definition}
An infinite word $w$ is \emph{abelian-square-rich} if there exists a positive constant $C$ such that for every $n$ one has 
\[\frac{1}{p_w(n)}\sum_{v\in \Fact(w)\cap\Sigma^n} \as{v} \geq C n^2.
\]
\end{definition}

Christodoulakis et al. \cite{Ch14} proved that a binary word of length $n$ contains  $\Theta(n\sqrt{n})$  distinct abelian-square factors on average; 
hence a random infinite binary word is almost surely not  abelian-square-rich.

In an abelian-square-rich word the number of distinct abelian squares contained in any factor is, on average, quadratic in the length of the factor. A stronger condition is that \emph{every} factor contains a quadratic number of distinct abelian squares:

\begin{definition}
An infinite word $w$ is {\it uniformly abelian-square-rich\/} if 
there exists a positive constant $C$ such that $\as{v} \geq C |v|^2$
for all $v \in \Fact(w)$.
\end{definition}

Clearly, if a word is uniformly abelian-square-rich, then it is also abelian-square-rich, but the converse is not always true. However, in the case of linearly recurrent words, the two definitions are equivalent. Moreover,  a uniformly abelian-square-rich word is always $\beta$-free for some $\beta$~\cite{DBLP:journals/tcs/FiciMS17}. Examples of uniformly abelian-square-rich words are the Thue--Morse word and the Fibonacci word.

In the opposite direction, one can ask what is the minimum number of abelian square factors in a word of length $n$. Let $f_d(n)$ be the least number of distinct abelian square factors in a word of length $n$ over $\Sigma_d$. By a result of Ker\"anen, $f_4(n)=0$ for every $n$ (see  Theorem~\ref{avoid_global}). Rao and Rosenfeld  proved that $f_3(n)\leq 34$ for every $n$ (see  Theorem~\ref{thm:avoidance} below). For binary words,  Entringer, Jackson and Schatz~\cite{EJS74} proved that every binary word of length $n^2+6n$ contains an abelian square of length $2n$, hence $f_2(n)$ is unbounded. 
The following conjecture  is supported by computer experiments:

\begin{conjecture}\cite{FiSa14}
Every binary word of length $n$ contains at least $\lfloor n/4 \rfloor$ distinct abelian square factors. That is, $f_2(n)=\lfloor n/4 \rfloor.$
\end{conjecture}
Should this conjecture be true, the bound is realized by  words of the form $0^{\lfloor n/2 \rfloor}10^{n-\lfloor n/2 \rfloor-1}$.

\bigskip

Abelian square factors give a characterization of a property related to shuffling. For finite words $u$ and $v$ we say that $v$ is a shuffle of $u$ with its reversal $u^R$ if there exist sequences of finite words $(U_i)_{i=0}^{n}$ and  $(V_i)_{i=0}^{n}$ such that $v=\prod_{i=0}^{n} U_iV_i$, $u=\prod_{i=0}^{n} U_i$, $u^R=\prod_{i=0}^{\infty} V_i$.  The following proposition gives a necessary condition for a word to be a shuffle of another word with its reversal:

\begin{theorem} \cite{DBLP:journals/eatcs/HenshallRS12} A binary word $v$ is an abelian square if and only if there exists a word $u$ such that $v$ is a shuffle of $u$ with its reversal $u^R$.\end{theorem}

Moreover, the ``if'' direction holds for arbitrary alphabets. However, there exist counterexamples for the ``only if'' part:  the word $012012$ is an example of a ternary abelian square that cannot
be written as the shuffle of a word with its reversal~\cite{DBLP:journals/eatcs/HenshallRS12}.

\subsection{Abelian antipowers}

Opposite to the notion of $k$-power, there is the notion of  \emph{$k$-antipower}\cite{DBLP:journals/jct/FiciRSZ18}. A $k$-{\emph{antipower}}, or antipower of order $k$, is a word of the form $v_1v_2\cdots v_k$ where all $v_i$'s have the same length and are pairwise distinct. For example, $001000111010$ is a $4$-antipower.

Fici, Restivo, Silva and Zamboni proved the following result:

\begin{theorem}\cite{DBLP:journals/jct/FiciRSZ18}\label{antipowers}
Every infinite word contains powers of every order or antipowers of every order.
\end{theorem}

By \cref{boundedAC}, we have that if a word has bounded abelian complexity, then it cannot contain abelian powers of every order, so in particular cannot contain powers of every order, therefore by  \cref{antipowers} it must contain antipowers of every order.

The abelian counterpart of an antipower is an \emph{abelian antipower}. An abelian $k$-antipower, or abelian antipower of order $k$, is a word of the form $v_1v_2\cdots v_k$ such that all $v_i$'s have the same length and pairwise distinct Parikh vectors.  For example, $010011$ is an abelian antisquare and an abelian anticube.  It is an open question whether~\cref{antipowers} can be generalized to abelian antipowers: 

\begin{problem}
 Does every infinite word contain abelian powers of every order or abelian antipowers of every order?
\end{problem}

Notice that if a word contains abelian antipowers of every order, then it must have unbounded abelian complexity.

By~\cref{paperfoldingpowers}, all paperfolding words contain abelian powers of every order. It has been proved that all paperfolding words also contain abelian antipowers of every order:

\begin{theorem}\cite{DBLP:journals/aam/FiciP019}
 All paperfolding words contain abelian $k$-antipowers for every $k>1$.
\end{theorem}

\section{Abelian avoidability}\label{sec:avoidability}

In this section we give a short overview of results and problems related to abelian avoidability. We do not go into details due to two recent excellent book chapters on abelian avoidance and related questions \cite{Ochem2018,RaSha_chapter}. 

\subsection{Avoidability of abelian powers}

Avoidability of powers and patterns is a well-studied area in
combinatorics on words. In this subsection, we provide some
results on avoidability of abelian powers. The study of abelian avoidance
started with a question of Erd\H{o}s, who asked whether it is
possible to construct an infinite word containing no abelian
square factor \cite{Erdos}. 

A $k$-power is a particular case of an abelian $k$-power. So,
unavoidability of $k$-powers implies unavoidability of  abelian
$k$-powers (but not vice versa).  So, for example, since every sufficiently long binary word contains a square, it is not possible to construct infinite binary words without abelian squares. Notice that by~\cref{boundedAC}, if a word avoids abelian powers, then it must have unbounded abelian complexity.

The
following theorem gives the minimal sizes of the alphabet for avoiding abelian
powers:

\begin{theorem}\label{avoid_global} \cite{Dekking1979181,Ker92} 

1. There exists an infinite word over an alphabet of size $4$ with
no abelian square factor.

2. There exists an infinite ternary word over with no abelian cube
factor.

3. There exists an infinite binary word with no abelian $4$-power factor.

The sizes of the alphabets are optimal.
\end{theorem}

The first statement of the theorem has been proved by Ker\"{a}nen in
1992 \cite{Ker92}, who improved a previous bound of $5$ given by  Pleasants~\cite{Pleasants} and the first bound of $25$ given by Evdokimov~\cite{Evdokimov}; the  two other statements by Dekking in 1979
\cite{Dekking1979181}. It is worth mentioning that Ker\"{a}nen's result relies on computer verification, while the two results by Dekking have a short elegant proof. We also refer to \cite{DBLP:journals/tcs/Keranen09} for more on abelian square-free morphisms and to \cite{DBLP:journals/ijac/Carpi93} for a shorter proof of item 1 in the theorem.

Moreover, it is known that the number of abelian square-free words of length $n$ on a four-letter alphabet grows exponentially in $n$ \cite{DBLP:journals/dam/Carpi98}. 
The same is true for ternary abelian-cube-free and binary abelian-4-free languages
\cite{ACR04,DBLP:journals/tcs/Currie04}.

The summary
of results on avoidability of (abelian) $k$-powers is provided in
Table \ref{table1}.

\begin{table}
\begin{center}\begin{tabular}{c|c|c}
 & usual & abelian \\ \hline
 squares & 3 & 4
 \\ \hline
 cubes & 2 & 3  \\ \hline
 4-powers & 2 & 2
\end{tabular}
\end{center}
\caption{Minimal sizes of the alphabets over which the
 corresponding powers are avoidable.} \label{table1}
\end{table}

To prove the avoidability results, it is enough to construct a word avoiding the corresponding power. An example of 
 an infinite word over  $\Sigma_4$ with no abelian square is given by a fixed point of the $85$-uniform substitution
 \[\begin{small}
 \psi:
\begin{cases}
 0\mapsto 0120232123203231301020103101213121021232021013010203212320231210212320232132303132120\\
 1 \mapsto  1231303230310302012131210212320232132303132120121310323031302321323031303203010203231\\
 2 \mapsto  2302010301021013123202321323031303203010203231232021030102013032030102010310121310302  \\
 3 \mapsto  3013121012132120230313032030102010310121310302303132101213120103101213121021232021013   \\
\end{cases}\end{small}
 \]
where the image of the letter $i$ is obtained from the image of the letter $i-1$ by adding $1$ to each letter modulo $4$.

The example of  a ternary word with no abelian cube factor is the fixed point of the  substitution 
 \[
 \psi':
\begin{cases}
 0\mapsto 0012\\
 1 \mapsto  112\\
 2 \mapsto  022  \\
\end{cases}
 \]

The example of a binary word with no abelian $4$-power factor can also be constructed as the fixed point of a substitution:

 \[
 \psi'':
\begin{cases}
 0\mapsto 011\\
 1 \mapsto  0001\\
\end{cases}
 \]

It is easy to see the optimality for the size of the alphabet: indeed, one can simply show, for example using a search tree, that there are only finitely many words without corresponding abelian powers. For example, for the three-letter alphabet we have the following:

\begin{proposition}
 Every ternary word of length $8$ contains an abelian square. 
\end{proposition}


To prove that abelian cubes are not avoidable over a binary alphabet, one has simply to increase the length:

\begin{proposition}
 Every binary word of length $10$ contains an abelian cube.
\end{proposition}

For more on constructions of abelian power-free words we refer to paragraph 4.6 in \cite{RaSha_chapter}.
For avoiding abelian powers and their generalizations see \cite{Ochem2018}.

Although abelian squares are unavoidable over a binary alphabet, one can ask whether
it is possible to construct an infinite binary word containing
only a finite number of abelian squares (as in the case of
ordinary squares, where there exists an infinite binary word
containing only $00$, $11$ and $0101$ as square factors). The answer to
this question is known, and it is negative; however, in the
ternary case, it is possible to construct infinite words
containing only a finite number of abelian squares:

\begin{theorem}\cite{EJS74,DBLP:journals/siamdm/RaoR18}\label{thm:avoidance}
 The following holds true:
\begin{enumerate}
\item Every infinite binary word contains arbitrarily long abelian squares.
 \item There exists an infinite ternary word with no abelian square of length $12$ or greater.
\end{enumerate}
\end{theorem}

We refer to \cite{EJS74} for a proof of the first part of the theorem, and to
\cite{DBLP:journals/siamdm/RaoR18} for the second part. The ternary word showed in \cite{DBLP:journals/siamdm/RaoR18} can be obtained by applying the morphism
  \[
g:
\begin{cases}
 0\mapsto 1110010002\\
 1 \mapsto 1220222122\\
 2 \mapsto 2222111212\\
 3 \mapsto 2222222200\\
 4 \mapsto 1111120100\\
 5 \mapsto 0000000100\\
\end{cases}
 \]
  to the fixed point of the substitution
 \[
h:
\begin{cases}
 0 \mapsto  024\\
 1 \mapsto  035\\
 2 \mapsto  135  \\
 3 \mapsto  132  \\
 4 \mapsto  054  \\
 5 \mapsto  124  \\
\end{cases}
 \]
This word contains precisely $34$ distinct abelian squares, the longest of which has length $10$.
 
The following conjecture is believed to be true, but is still
unproved:

\begin{conjecture}[M\"{a}kel\"{a}, \cite{keranen1}]
 There exists an infinite ternary word whose only abelian squares are $00$, $11$, $22$.
\end{conjecture}



Another conjecture stated by M\"{a}kel\"{a} was that there exists an
infinite binary word containing only $000$ and $111$ as abelian
cube factors, but this has been shown to be false in \cite{DBLP:journals/moc/RaoR16}.
However, the following modification of M\"{a}kel\"{a}'s question is still
open:

\begin{problem}
 Is it possible to construct an infinite binary word containing only a finite number of abelian cubes?
\end{problem}

Finally, Peltom\"aki and Whiteland~\cite{DBLP:journals/aam/PeltomakiW20} considered \emph{cyclic abelian avoidance}. A finite word $w$ avoids abelian $k$-powers cyclically if for each
abelian $k$-power of period $m$ occurring in the infinite word $w^\omega$, one has $m\geq |w|$. For example, let $w = 1000100$. Then both $w$ and $w^2$ avoid abelian $5$-powers. However, the
word $w^3$ has the abelian $5$-power $100 \cdot 010 \cdot 010 \cdot 001 \cdot 001$ of period 3 as a prefix. Therefore, $w$ does not avoid abelian $5$-powers cyclically. It does not avoid abelian $6$-powers cyclically either, since $w^4$
contains an abelian $6$-power of period $4$
beginning from the second letter. However, it avoids abelian $7$-powers cyclically~\cite{DBLP:journals/aam/PeltomakiW20}. Let $A(d)$ be the least integer $k$ such that for all $n$ there exists a word of length $n$
over a $d$-letter alphabet that avoids abelian $k$-powers cyclically. Similarly, let $A_\infty(d)$ be
the least integer $k$ such that there exist arbitrarily long words over a $d$-letter alphabet
that avoid abelian $k$-powers cyclically.

\begin{theorem}\cite{DBLP:journals/aam/PeltomakiW20}
One has $5 \leq A(2) \leq 8$, $3 \leq A(3) \leq 4$, $2 \leq A(4) \leq 3$, and $A(d) = 2$ for every $d \geq 5$.

Moreover, $A_\infty(2) = 4$, $A_\infty(3) = 3,$ and $A_\infty(4) = 2.$
\end{theorem}

\subsection{Avoiding fractional abelian repetitions and other generalizations of abelian powers}

In the classical (non-abelian) sense a fractional repetition is defined as a word of the form $w^nv$, $n>0$, where $w$ is primitive and $v$ is a prefix of $w$. The exponent of the repetition is then $n+\frac{|v|}{|w|}.$  For example, the word $0010010$ has exponent $7/3$ so it is a $7/3$-power.

For a $d$-letter alphabet ($d \geq  2$), the repetition threshold is the number $RT(d)$
which separates $d$-unavoidable and $d$-avoidable repetitions. For example, the Thue--Morse word shows that $RT(2)=2$. The famous Dejean’s
conjecture dating back to 1972 \cite{DBLP:journals/jct/Dejean72} stated that $RT(3) = 7/4$, $RT(4) = 7/5$, and $RT(d) = d/(d - 1)$
for every $d>5$. The conjecture has been proved in a series of papers --- the last cases have been proved independently by Rampersad and Currie~\cite{DBLP:journals/moc/CurrieR11}, and Rao~\cite{DBLP:journals/tcs/Rao11}. 

In analogy with avoiding of fractional powers, one can wonder whether one can avoid fractional abelian powers. 

\begin{theorem}\cite{DBLP:journals/ejc/CassaigneC99} Let $\beta$ be a real number, $1<\beta <2$. There exists an infinite word over a finite alphabet which contains no factor of the form $xyz$ with $|x yz|/|x y|\geq \beta$ and where $z$ is abelian equivalent to $x$. 
\end{theorem}

This kind of factor can be regarded as a fractional abelian power of exponent $\beta$. For example, $01110$ has abelian exponent $\frac{5}{3}$ in this sense, with $x=01$, $y=1$, $z=10$. 

There are several other natural generalizations of the notion of a fractional power to the abelian case. For two Parikh vectors $\PV(u)$ and $\PV(v)$, we write $\PV(u)\subseteq \PV(v)$ if $\PV(u)$ is component-wise smaller than or equal to $\PV(v)$. A word $uv$ is called an \emph{abelian inclusion} if $\PV(u)\subseteq \PV(v)$. Consider a word of the form $w = w_1 \cdots w_m v$, where $w_1\sim_{ab} \cdots \sim_{ab} w_m$, and $\PV(v)\subseteq \PV(w_1)$ (hence $\PV(v)\subseteq \PV(w_i)$ for every $i$). A word of this form can be considered as a fractional abelian repetition of exponent $m+\frac{|v|}{|w_1|}$. 

In  \cite{DBLP:journals/ita/SamsonovS12}, three versions of the notion of fractional abelian repetition are considered: in a weak form, i.e., without additional restrictions; in a strong form, i.e., with a requirement that $\Pref_{|v|}(w_1)\sim_{ab} v$; and in a semi-strong form, i.e., with a requirement that $\mathcal{P}(v)\subseteq \bigvee_{i=1}^m \mathcal{P}(\Pref_{|v|}(w_i))$, where $\bigvee$ is the operation of taking the maximum componentwise.  The authors found lower and upper bounds for abelian repetition thresholds, some of which are conjectured to be tight. 

In \cite{ DBLP:journals/jalc/AvgustinovichF02}, the authors considered avoiding abelian inclusions. For two words $u$ and $v$, we say that $v$ {\emph{majorizes}} $u$ if for each letter $a\in \Sigma$, $|u|_a\leq |v|_a$, i.e., if $\PV(u)\subseteq \PV(v)$.
Let us fix a function $f(l): \mathbb{N} \to \mathbb{R}$ and call a word $w = uv$ an $f(l)$-inclusion if
$v$ majorizes $u$ and $|v| \leq |u|+f(|u|)$. As usual, we say that a word avoids $f(l)$-inclusions if none of its factors is an $f(l)$-inclusion.

\begin{theorem}{\cite{ DBLP:journals/jalc/AvgustinovichF02}} For every arbitrarily small
constant $c>0$, $cl$-inclusions are unavoidable.\end{theorem}

\begin{theorem}{\cite{ DBLP:journals/jalc/AvgustinovichF02}} For every arbitrarily large
constant $n$, there exists a word on $4(n + 1)$
letters avoiding $n$-inclusions.\end{theorem}

\subsection{Abelian pattern avoidance}


 For two words $P$ and $w$, we say that
$w$ avoids the pattern $P$ if there is no non-erasing morphism $h$ such that $h(P)$ is a factor of $w$,
or equivalently if there is no factor $w_1w_2 \cdots w_{|P|}$
in $w$ such that for every $i$ and $j$ $P_i = P_j$ implies $w_i = w_j$.

Abelian pattern avoidance in defined similarly to usual pattern avoidance. Let $P = P_1P_2 \cdots P_n$ be a pattern, where the $P_i$ are letters. Then we say that a word $w \in \Sigma_d^*$
\emph{realizes $P$ in the abelian sense} if there exist $w_1, \dots, w_n \in \Sigma_d^+$  such that $w = w_1w_2 \cdots w_n$ and
for every $i$ and $j$ $P_i = P_j$ implies $w_i \sim_{ab} w_j$.  

We say that a pattern is \emph{$d$-avoidable} (resp., \emph{$d$-abelian avoidable}) if it is avoidable (resp., abelian avoidable) over $\Sigma_d$.


Pattern avoidance in the usual sense is a well-studied topic. There is an explicit characterization of patterns that are avoidable in the usual sense (Bean, Ehrenfeucht, McNulty  \cite {BEM79}, and independently Zimin~\cite{Zimin84}); see also Chapter 3 in  \cite{LothaireAlg}. However, the problem of finding the avoidability index of a pattern, i.e., the minimal size of the alphabet for which it is avoidable, is still unsolved, 
and not as much is known about  avoidability of abelian patterns.
For example, it has been shown in \cite{DBLP:journals/tcs/CurrieV08} that all long enough binary abelian patterns are  2-abelian avoidable, and the bound has been improved in \cite{DBLP:conf/mfcs/Rosenfeld16}:

\begin{theorem}\cite{DBLP:conf/mfcs/Rosenfeld16} Binary patterns of length greater than 14 are  2-abelian avoidable.\end{theorem}

The best known lower bound is 7 \cite{DBLP:conf/mfcs/Rosenfeld16}.
A similar fact has been proved for avoidance over a three-letter alphabet:

\begin{theorem}\cite{DBLP:conf/mfcs/Rosenfeld16} Binary patterns of length greater than 8 are 3-abelian avoidable.\end{theorem}

It is easy to see that all binary patterns except for short ones ($A$, $AB$ and $ABA$, up to renaming letters) are avoidable over $4$ letters. This follows from the fact that abelian squares are avoidable over $4$ letters, and all other binary patterns must contain a square.

In \cite{CurrieLinek2001}, the authors classify ternary patterns which are abelian avoidable. As in the ordinary case, the problem of determining whether a given pattern is avoidable in the abelian sense over an alphabet of a given size is yet unsolved. Moreover, no algorithm is known, even if we do not restrict the size of the alphabet, although in the ordinary sense the solution is given by Zimin algorithm~\cite{Zimin84}.


Since words avoiding patterns in the abelian or in the usual sense are often constructed as fixed points of substitutions, it is reasonable to consider the following decision problem: Given a substitution $h$ with an infinite fixed point $w$ and an integer $k\geq 2$, determine if $w$ is (abelian) $k$-power free. The decidability of this problem for usual powers has been studied in several papers and proved in general by Mignosi and S\'e\'ebold~\cite{DBLP:conf/icalp/MignosiS93}. Currie and Rampersad showed that the problem is also decidable for abelian powers in the case of morphisms satisfying certain conditions \cite{CURRIE2012942}. The result has been further generalized in \cite{DBLP:journals/siamdm/RaoR18} and \cite{DBLP:conf/mfcs/Rosenfeld16} for wide classes of patterns and other types of repetitions.

The related problem of determining if a morphism is $k$-power free (i.e., maps $k$-power free words to $k$-power free words) has also been examined previously. This is not quite the same question as the one posed above, since it is possible for a morphism to generate a $k$-power free word without being $k$-power free. Carpi gave sufficient conditions for a morphism to preserve abelian $k$-power freeness, which is
conjectured to be a characterization \cite{DBLP:journals/ijac/Carpi93}.


\section{Abelian periods and borders}\label{sec:borders}

The notion of a period can be naturally generalized to the abelian case, and many classical results on periodicity are generalized to the abelian case. However, in some cases the problem becomes harder (or easier!), and sometimes there is no clear generalization.

There are several possible ways to define an abelian period of a word: either we can require a period to start from the very beginning of the word, or we can admit a preperiod. Depending on the question, one or another definition is more natural. In this section we make a survey of abelian versions of some classical results on abelian periods, such as the Fine and Wilf lemma, primitive words, the Critical Factorization theorem and some others.

\subsection{Abelian versions of classical periodicity theorems}

In this subsection we discuss how classical periodicity theorems (the Fine and Wilf periodicity lemma and the Critical Factorization theorem) can be generalized to the abelian setting.

Constantinescu and Ilie~\cite{CI2006} introduced the following generalization  of the notion of a period of a finite word to the abelian case. Recall that for a nonempty word $u$ over a fixed ordered alphabet we let $\PV(u)$ denote its Parikh vector.  We let $|\PV(u)|$ denote the norm of $\PV(u)$, that is, the sum of its components. We further write $\PV(u)\subset \PV(v)$ if $\PV(u)$ is component-wise smaller than or equal to $\PV(v)$ and $|\PV(u)|<|\PV(v)|$.


%

\begin{definition}\label{def:abper}
A word $w$ has an \emph{abelian period} $p$, with preperiod $h$, if
$w=u_0u_1 \cdots u_{m-1}u_m$ for some words $u_0, \dots , u_m$ such that:

\begin{itemize}
 \item $\PV(u_{0})\subset \PV(u_{1})=\cdots =\PV(u_{m-1})\supset \PV(u_{m})$,
 \item $|\PV(u_{0})|=h$, $|\PV(u_{1})|=p$.
\end{itemize}

\end{definition}

The words $u_0$ and $u_m$ are called resp.\ the \emph{head} and the
 \emph{tail} of the abelian period.
 Notice that the length $t=|u_m|$ of the tail is uniquely determined
 by $h$, $p$ and $|w|$, namely $t=(|w|-h) \bmod p$.


The following lemma gives an upper bound on the number of distinct
pairs $(p,h)$ of abelian periods with preperiods
 of a word:

\begin{lemma}
\label{lemma-max} A word of length $n$ can have $\Theta(n^2)$ different pairs
$(p,h)$ of abelian periods with preperiods.
\end{lemma}

\begin{proof}
For every $d$, the word $w=(12\cdots d)^{n/d}$ has abelian period
 $p$ with preperiod $h$ for any $p\equiv 0 \bmod d$ and every $h$ such that
 $0 \leqslant h \leqslant \min(p-1,n-p)$. Therefore, $w$ has $\Theta(n^2)$ different pairs $(p,h)$ of the lengths of abelian periods with prepriods.
\end{proof}

Often, we are only interested in the integer $p$ and not in the
length $h$ of the head.

Let us recall the following classical result dating back to 1965, known as the Periodicity Lemma or Fine and Wilf's Lemma.

\begin{lemma}[\cite{FW65}]\label{FW}
 Let $w$ be a word. If $p$ and $q$ are periods of $w$ and $|w|\geq p+q-\gcd(p,q)$, then $\gcd(p,q)$ is a period of $w$.
\end{lemma}

The value $p+q-\gcd(p,q)$ in the statement of~\cref{FW} is optimal, in the sense that for any $p$ and $q$ it is
possible to construct a word with periods $p$ and $q$ and length
$|w|=p+q-\gcd(p,q)-1$ such that $\gcd(p,q)$ is not a period of $w$. In fact, a word is called \emph{central} if it has two coprime periods $p$ and $q$ and length equal to $p+q-2$. For example, $010$ and $010010$ are central words. Every central word is a binary palindrome (but there are binary palindromes that are not central). Moreover, central words are rich.

We now present a generalization of the Fine and Wilf's lemma to the case of abelian periods.

Let $\alph(w)$ be the set of distinct letters appearing in $w$.
By~\cref{FW}, if a word $w$ has two coprime periods
$p$ and $q$ and length $|w|\geq p+q-1$, then $|\alph(w)|=1$.

\begin{theorem}{\cite{CI2006}}
 If a word $w$ has coprime abelian periods $p$ and $q$ and length $|w|\geq 2pq-1$, then $|\alph(w)|=1$, that is, $w$ is a power of a single letter.
\end{theorem}

The latter result has been generalized by Simpson to the case when the abelian periods $p$ and $q$ are not coprime:

\begin{theorem}{\cite{DBLP:journals/tcs/Simpson16}}\label{th:simpson}
 If a word $w$ has abelian periods $p=p'd$ and $q=q'd$ and length $|w|\geq 2p'q'd-1$ for integers $d$, $p'$, $q'$, then $|\alph(w)|\leq d$. 

 Moreover, if the difference $||v_0|-|u_0||$ of the lengths of the heads of the two periods  $p$ and $q$ is not a multiple of $d$, then the previous bound  can be reduced to $2p'q'd-2$.
\end{theorem}

\begin{example}
 Let $w=010201001201020102001$ of length $21$. Since $w$ can be factored as
\begin{align*}
w & = u_0u_1u_2u_3u_4u_5 = 010\cdot 2010\cdot 0120\cdot 1020\cdot 1020\cdot 01 \\
&  = v_0v_1v_2v_3 = 0102 \cdot 010012 \cdot 010201 \cdot 02001
\end{align*}
it follows that $w$ has abelian periods $4=2\cdot 2$ and $6=3\cdot 2$, and we have $||v_0|-|u_0||=1$, which is not a multiple of $d=2$. One can see that $w$ cannot be extended to the left nor to the right keeping the same abelian periods with this factorization.
Nevertheless, $w$ can also be factored as
\begin{align*}
w & = v'_0v'_1v'_2v'_3 = 01020\cdot 100120 \cdot 102010 \cdot 2001
\end{align*}
and now $||v'_0|-|u_0||=2=d$. One can verify that with these factorizations $w$ can be extended to the right with the letter $0$ keeping the abelian periods $4$ and $6$, resulting in a word of length $22=2\cdot 3\cdot 2 -2$. In accordance with Theorem \ref{th:simpson}, the word $w0$ cannot be extended to the left nor to the right to a word of length $23=2\cdot 3\cdot 2 -1$ having abelian periods $4$ and $6$.
\end{example}

Interestingly enough, in the classical version, the Fine and Wilf’s theorem basically says that if a word has two periods $p$ and $q$ and is long enough, then it also has period $\gcd(p,q)$. This fact cannot be extended to
abelian periods which are not relatively prime. That is, if $\gcd(p, q) = d > 2$, then
the two abelian periods $p$ and $q$ cannot impose the abelian period $d$, no matter how
long the word is. In~\cite{CI2006}, the authors exhibited an infinite word,  $w = (001110100011)^{\omega}$,
which has abelian periods 4 and 6, but not 2.

We now discuss abelian versions of another classical periodicity result, a central factorization theorem. This result relates global periodicity of a word with its local periods, defined as the length of the shortest square centered at each position. This relation can be stated for finite, infinite or biinfinite words (a biinfinite word is a map from $\Z$ to $\Sigma_d$). For example, for biinfinite words the following holds:

\begin{theorem}[\cite{CV78}] A biinfinite word $x$ is periodic if and only if there exists an integer
$l$ such that $x$ has at every position a centered square with period at most $l$.\end{theorem}

A similar result holds for powers to the left of each position, although in this case a square is not enough to guarantee periodicity, but the threshold is given by the golden ratio:

\begin{theorem}[\cite{DBLP:journals/tcs/MignosiRS98}]  A right-infinite word $x$ is ultimately periodic if and only if there
exists $n_0$ such that for every $n \geq  n_0$ the word $\Pref_n(x)$ has a $\varphi^2$-suffix, where
$\varphi = (1+\sqrt{5})/2$.\end{theorem}

This bound is optimal; an example of an aperiodic word with $(\varphi^2-\varepsilon)$-suffix at each position is given by the Fibonacci word.

These properties do not seem to generalize well for abelian powers. In particular, for each $k$, there exist aperiodic words with an abelian $2k$-power centered at each position:

\begin{theorem}\cite{AKP2012} For every integer $k$, there exists a bi-infinite aperiodic word with an abelian $2k$-power with period of length at most $2(k+1)^2$ centered at each position. \end{theorem}

An infinite word $x$ is called \emph{abelian 
periodic} if $x = v_0 v_1 \cdots$, where $v_k\in\Sigma_d^*$ for
$k\geq0$, and $v_i \sim_{ab} v_j$ for all integers $i, j \geq 1$; or \emph{abelian aperiodic} otherwise. There exist words that are not abelian periodic, but contain a centered abelian square of bounded length at each position \cite{DBLP:journals/ejc/CharlierHPZ16}. Consider the  family of infinite words of the following form:
$$(000101010111000111000(111000)^*111010101)^{\omega}$$
where $w^*$ denotes zero or more repetitions of $w$ and $w^\omega=www\cdots$ denotes an infinite concatenation of copies of $w$.
Words of this form have an abelian square of length at most
$12$ at each position. It is not hard to see that this family contains abelian aperiodic words.

\subsection{Abelian primitive words}

An abelian $k$-power is a nonempty word of the form $w=w_1w_2\cdots w_k$, where all $w_1,w_2,\ldots, w_k$ have the same Parikh vector. A word is called {\emph{abelian primitive}} if it is not an abelian $k$-power for any $k$. A word $w$ has an \emph{abelian root} $u$ if $u$ is a prefix of $w$ and $w$ is an abelian $|w|/|u|$-power. If $u$ is an abelian root of length $\ell$ of a word $w$ of length $n$, then clearly $w$ has also abelian roots of length $\ell'$ for each $\ell'$ multiple of $\ell$ that divides $n$. If $u$ is abelian primitive, then it is called an \emph{abelian primitive root}. Recall that in the classical case the primitive root of a word is unique. 
On the contrary, in the abelian case a word can have more than one abelian primitive root. Indeed, the example from the previous section (due to \cite{CI2006}) gives  an infinite word with two distinct abelian periods not dividing each other, namely  $w = (001110100011)^{\omega}$ which has abelian periods $4$ and $6$.
The situation has been studied in~\cite{DBLP:journals/ijfcs/DomaratzkiR12}, where it has been proved that if $u$ and $v$ are distinct abelian primitive roots of the same word, then $\gcd(|u|,|v|)\geq 2$. The authors also gave upper and lower bounds on the number of distinct abelian primitive roots of a word.

Another natural question is related to the generalization of the classical Lyndon--Sch\"utzenberger lemma:

\begin{lemma}\label{LS}
 Let $u,v$ be two words. Then $uv=vu$ if and only if $u$ and $v$ have the same primitive root.
\end{lemma}

Let us write $u \approx_n v$ if $u$ and $v$ can be decomposed in the same number of contiguous blocks of length $n$ all having the same Parikh vector. For example, $012 021 012 \approx_3 210 120 120$. The following generalization of Lemma~\ref{LS} has been proved in~\cite{DBLP:journals/ijfcs/DomaratzkiR12}:

\begin{lemma}
 Let $u,v$ be two words such that $uv \approx_n vu$. If $u$ has an abelian primitive root of length $n$, then $v$ does as well, and these abelian primitive roots are the same.
\end{lemma}

Finally, in~\cite{DBLP:journals/ijfcs/DomaratzkiR12} it has been proved that the language of abelian primitive words is not context-free, while an analogous result for primitive words is a longstanding open question~(see~\cite{Primitive}).

\subsection{Abelian borders}\label{subsec:borders}

A finite word is called \emph{bordered} if it has a border, i.e., a
proper prefix which is also a suffix, and \emph{unbordered}
otherwise. A natural generalization is therefore: a finite word has an \emph{abelian border} if it has a proper prefix that is abelian equivalent to the suffix of the same length. If a word does not have any abelian border, it is called \emph{abelian unbordered}. Of course, if a word has a border then it has an abelian border, but there exist unbordered words having an abelian border, e.g. the unbordered abelian square $00110101$. Clearly, a word of length $n$ has an abelian border of length $\ell\leq n/2$ if and only if it has an abelian border of length $n-\ell$.

Remember that if a word $w$ has a border of length $\ell$, then $|w|-\ell$ is a period of $w$. With the definition of abelian period given in Definition~\ref{def:abper}, it is not always true that if $w$ has an abelian border of length $\ell$, then $|w|-\ell$ is an abelian period of $w$.

In~\cite{DBLP:journals/ijfcs/GocRRS14}, the authors counted binary abelian bordered words via a bijection with irreducible symmetric Motzkin paths. Besides that, the lengths of the abelian unbordered factors occurring in the Thue–Morse word are characterized using the automatic theorem-proving tool Walnut, a software package that implements a mechanical decision procedure for deciding certain combinatorial properties of automatic sequences. We refer to the recent book of J. Shallit for more results obtained using Walnut \cite{shallit_2022}.

Concurrently and independently of~\cite{DBLP:journals/ijfcs/GocRRS14},  in~\cite{DBLP:journals/dam/ChristodoulakisCCI14} the authors proved the following result:

\begin{theorem}\cite{DBLP:journals/dam/ChristodoulakisCCI14}
The number of binary words of length $n$ with shortest abelian border of length $k$ is $\Theta(\frac{2^n}{k\sqrt{k}})$. In fact, that number is $2\sqrt{2}\frac{2^n}{k\sqrt{\pi k}}+o(\frac{2^n}{k\sqrt{\pi k}})$.
\end{theorem}

The exact number, however, has been recently found by Blanchet-Sadri, Chen and Hawes:

\begin{theorem}\cite{Sadri22}
The number of binary words of length $n$ with shortest abelian border of length $k$ is $2^{n-2k+1}\cdot \frac{1}{n}\binom{2n-2}{n-1}$.
\end{theorem}

 A classical result of Ehrenfeucht and  Silberger \cite{EHRENFEUCHT1979101} gives a relation between  periodicity and bordered factors; it states that 
 an infinite word is purely periodic
    if and only if it contains only finitely many unbordered
    factors:

    \begin{theorem}\emph{\cite{EHRENFEUCHT1979101}} \label{th:borders}
    An infinite word $x$ is purely periodic if and only if
    there exists a constant $C$ such that every factor $v$ of $x$ with $|v|\geq C$ is bordered.
\end{theorem}

If we replace periodic with abelian periodic, 
 an analogous assertion does not hold: abelian
periodicity does not imply a finite number of unbordered
factors. For example, the Thue--Morse word has abelian period 2, but contains unbordered factors of unbounded lengths since it is aperiodic. If we replace borders with abelian borders, the reciprocal does not hold even in a  stronger form: even if all long factors have short abelian  borders, the word does not have to be periodic.

\begin{proposition}\cite{DBLP:journals/ejc/CharlierHPZ16} There exist an infinite aperiodic word $x$ and constants $C$, $D$ such that every factor $v$ of $x$ with $|v| \geq C$ has an abelian border of length at most $D$.
\end{proposition}

Whether it holds for abelian periodicity is an open question:

\begin{problem}\cite{DBLP:journals/ejc/CharlierHPZ16}
 Let $x$ be an infinite word and $C$ a constant such that every factor v of $x$ with $|v| \geq C$ is
abelian bordered. Does it follow that $x$ is abelian periodic?
\end{problem}

However, there exists an abelian analogue of the following weaker
version of Theorem \ref{th:borders}:

\begin{theorem} Let $x$ be an infinite word having only finitely many unbordered factors. Then there exists a
constant $N$ such that $x$ contains at most $N$ factors of each given length $n \geq 1$. In other words, $x$ has bounded
factor complexity.
\end{theorem}

Notice that boundedly many unbordered factors  implies ultimate periodicity but not, in general, pure periodicity. Take for example the word $01^\omega$, which in fact has infinitely many unbordered factors.

\begin{theorem}\cite{DBLP:journals/ejc/CharlierHPZ16} Let $x$ be an infinite word having only finitely many abelian unbordered factors. Then there
exists a constant $N$ such that $x$ contains at most $N$ abelian equivalence classes of factors of each given
length $n \geq 1$. In other words, $x$ has bounded abelian complexity.\end{theorem}

See also Subsection \ref{subsec:WAP} for other generalizations of Theorem \ref{th:borders} in the abelian setting.

Abelian borders turn out to be a useful instrument for studying some  combinatorial properties of words which do not seem to be directly related at the first glance. For example, they give a necessary condition for a word to be self-shuffling. An infinite word $x$ is called self-shuffling if there exist sequences of finite words $(U_i)_{i=0}^{\infty}$ and  $(V_i)_{i=0}^{\infty}$ such that $x=\prod_{i=0}^{\infty} U_iV_i=\prod_{i=0}^{\infty} U_i=\prod_{i=0}^{\infty} V_i$. In other words, $x$ is a shuffle of two copies of itself. The following proposition gives a necessary condition for a word to be self-shuffling:

\begin{proposition} \cite{DBLP:journals/jct/CharlierKPZ14}
If $x$ is self-shuffling, then for every positive integer $N$ there exists a positive integer $M$ such that every prefix $u$ of $x$ with $|u| \geq M$ has an abelian border $v$ with $|u|/2 \geq |v| \geq N$. In particular, $x$ must begin in only a finite number of abelian unbordered words.\end{proposition}

We discussed shuffling in relation with abelian squares in Subsection~\ref{subsec:squares} .

Furthermore, abelian borders turn out to be useful for studying certain palindromicity properties. For example, in \cite{DBLP:journals/dam/HolubS09} a notion of a minimal palindromic word has been introduced. Let $w = w_1 \cdots w_{n}$ be a word of length $n$ over
$\Sigma_d$, and let $l\leq n$. Let $s: \mathbb{N} \to \mathbb{N}$ be an
increasing map such that $s(l) < n$. Then the word
$w_{s(1)}\cdots w_{s(l)}$ is a \emph{scattered subword} of length $l$ of
$w$. Clearly, every binary word  contains a palindromic scattered subword of length at least half of
its length -- a power of the prevalent letter. A  word is called \emph{minimal palindromic} if it contains a palindromic scattered subword longer than half of its length. Holub and Saari~\cite{DBLP:journals/dam/HolubS09} proved that minimal palindromic binary words are abelian unbordered.  
This has been recently generalized to any size of the alphabet  by Ago and Basic:

\begin{theorem}\cite{DBLP:journals/dam/AgoB21}
Minimal palindromic words are abelian unbordered. 
\end{theorem}

\section{Abelian properties of Sturmian words}\label{sec:sturmian}




A Sturmian word can be defined as an infinite word that has $n+1$ distinct factors of each length $n\geq 0$. There exists a vast literature on Sturmian words (see, e.g., Chapter 2 in \cite{LothaireAlg} for a presentation of the topic). We now give some basic notions that are needed to present the results on their abelian combinatorics.

An infinite aperiodic binary word is Sturmian if and only if it is balanced, in the sense of Definition~\ref{def:balance}. Therefore, for every  $n\geq 0$, the $n+1$ factors of length $n$ of a Sturmian word are partitioned in two abelian equivalence classes (often called \emph{light} factors and \emph{heavy} factors, depending on the number of $1$s they contain).  For example, the 5 factors of length 4 of the Fibonacci word are: $0010$, $0100$ (light factors), $0101$, $1001$ and $1010$ (heavy factors).

A useful description of Sturmian words is the following.
Given an irrational number $0<\alpha<1$ and a real number $\rho$, the Sturmian word $\underline{s}_{\alpha,\rho}$ (resp.,~$\overline{s}_{\alpha,\rho}$) with \emph{slope} $\alpha$ and \emph{intercept} $\rho$ is the infinite word 
\[
\underline{s}_n=\lfloor \alpha(n+1)+\rho \rfloor - \lfloor \alpha n+\rho \rfloor
\]
resp.,
\[
\overline{s}_n=\lceil \alpha(n+1)+\rho \rceil - \lceil\alpha n+\rho \rceil,
\]
for every $n\geq 0$. 
We let $\underline{I}_0$ denote the interval $[0,1-\alpha)$,  $\underline{I}_1$ the interval $[1-\alpha,1)$.  Denoting by $\{\theta\}$  the fractional part $\theta -\lfloor \theta\rfloor$ of a real number $\theta$, we have that for every $n\geq 0$
$$\underline{s}_{n} =
\left\{
	\begin{array}{ll}
		0  & \mbox{if } \{ \rho + n\alpha \}\in \underline{I}_0,\\
		1  & \mbox{if } \{ \rho + n\alpha \}\in \underline{I}_1.
	\end{array}
\right.$$ 
The same expression can be written for $\overline{s}_{n}$ with  $\overline{I}_0=(0,1-\alpha]$ and $\overline{I}_1=(1-\alpha,1]$. 

Notice that $\underline{s}_{\alpha,\rho}=\overline{s}_{\alpha,\rho}$ except when $\rho + n \alpha$ is an integer for some $n\geq 0$, that is, $\rho$ is congruent to $-n\alpha$ modulo $1$, in which case the two words differ at position $n$, and also at position $n-1$ if $n>0$. In particular, when $\rho=0$,  we have $\underline{s}_{\alpha,0}=0s_{\alpha,\alpha}$ and $\overline{s}_{\alpha,0}=1s_{\alpha,\alpha}$. If $\rho=\alpha$, the Sturmian word $s_{\alpha,\alpha}$ is called \emph{characteristic} or \emph{standard}. 

Recall that the (simple) \emph{continued fraction} of an irrational number $\alpha$, $0<\alpha<1$, is
\begin{equation}\label{cf}
  \alpha = \dfrac{1}{a_1 + \dfrac{1}{a_2 + \ldots}}
\end{equation}
and is usually denoted by its sequence of \emph{partial quotients} as follows: $\alpha=[0;a_{1},a_{2},\ldots ]$.
Each  finite truncation $[0;a_{1},a_{2},\ldots,a_{i}]$ is a rational number $p_{i}/q_{i}$ (we take $p_{i}$ and $q_{i}$ coprime) called the $i$-th \emph{convergent} to $\alpha$.  The sequence $(q_i)_{i\geq 0}$ can be defined  by: $q_{-1}=0$, $q_0=1$ and $q_n=a_nq_{n-1}+q_{n-2}$ for $n\geq 1$.
We say that $\alpha=[0;a_{1},a_{2},\ldots ]$ has bounded partial quotients if the sequence $(a_{i})_{i\geq 0}$ is bounded. 

For example, one has $\phi-1=[0;1,1,\ldots]$, where $\phi$ is the golden ratio, and the sequence $(q_i)_{i\geq 0}$ is the sequence of Fibonacci numbers.

The characteristic Sturmian word $s_{\alpha,\alpha}$ of slope $\alpha$, $0<\alpha<1$, can be obtained as the limit of the sequence of words $(s_n)_{n \ge 0}$ defined recursively as follows:
Let $[0; d_0+1, d_1,d_2, \ldots]$ be the continued fraction expansion of $\alpha$, and define
$s_{-1}=1$, 
$s_0=0$ and 
$s_{n+1}=s_{n}^{d_{n}}s_{n-1}$ for every $n\geq 0$.
Note that  $s_{\alpha,\alpha}$ starts with letter $1$ if and only if $\alpha > 1/2$, i.e., if and only if $d_0 = 0$. In this case, $[0; d_1 + 1, d_2, \ldots]$ is the continued fraction expansion of $1-\alpha$, and $s_{1-\alpha,1-\alpha}$ is the
word obtained from $s_{\alpha,\alpha}$ by exchanging $0$'s and $1$'s. 

The finite words $s_n$ are called \emph{standard words}. A standard word $s_n$, $n\geq 1$, is always of the form $s_n=c01$ or $s_n=c10$, where $c$ is a central word (recall from Sec.~\ref{sec:borders} that a central word is a word that has two coprime periods $p$ and $q$ and length equal to $p+q-2$). 
 


It is known that two Sturmian words have the same set of finite factors if and only if they have the same slope $\alpha$. Hence, in what follows, we will write $s_\alpha$ to denote any Sturmian word of slope $\alpha$.

\subsection{Abelian powers in Sturmian words}

Richomme, Saari and Zamboni~\cite{Richomme201179} proved that in every Sturmian
 word, for any position and for every positive integer $k$, there is
an abelian $k$-power starting at that position. 

Recall that  $\|\alpha\|$ denotes the distance between a real number $\alpha$ and the nearest integer, i.e., $\|\alpha\|=\min(\{\alpha\},\{-\alpha\})$. 
In~\cite{tcs16}, the following result is proved:

\begin{theorem}\label{the:main1}
  Let $s_\alpha$ be a Sturmian word of slope $\alpha$ and $m$ be a positive integer. Then $s_\alpha$ contains an
  abelian power of period $m$ and exponent $k \geq 2$ if and only if $\|m\alpha\| < \frac{1}{k}$. In particular,
  the maximum exponent $k_m$ of an abelian power of period $m$ in $s_\alpha$ is the largest integer $k$ such that
  $\|m\alpha\|< \frac{1}{k}$, i.e.,
  \begin{equation*}
    k_{m}=\left \lfloor \frac{1}{ \|m\alpha\| } \right \rfloor.
  \end{equation*}
\end{theorem}

\begin{example}
In Table \ref{tab:fici2} we give the first values of the sequence $k_{m}$ for the Fibonacci word $f$. We have $k_{2}=4$, since $\{2(\phi-1)\}\approx 0.236$, so the largest $k$ such that $\{2(\phi-1)\}< 1/k$ is $4$. Indeed, $10100101$ is an abelian power of period $2$ and exponent $4$, and the reader can verify that no factor of $f$ of length $10$ is an abelian power of period $2$.

For $m=3$, since $\{-3(\phi-1)\}\approx 0.146$, the largest $k$ such that $\{-3(\phi-1)\}< 1/k$ is $6$. Indeed, $001001010010010100$ is an abelian power of period $3$ and exponent $6$, and the reader can verify that no factor of $f$ of length $21$ is an abelian power of period $3$.
\end{example}

\begin{table}
\centering
\begin{small}
\begin{raggedright}
\begin{tabular}{c *{30}{@{\hspace{3.1mm}}c}}
$m$\hspace{2mm} & \textbf{1} & \textbf{2} & \textbf{3} & 4 & \textbf{5} & 6 & 7 & \textbf{8} & 9 & 10 & 11 & 12 & \textbf{13} & 14 & 15 & 16 & 17 & 18 & 19 & 20 & \textbf{21}
\\
\hline \\
$k_{m}$\hspace{2mm} & \textbf{2} & \textbf{4} & \textbf{6} & 2 & \textbf{11} & 3 & 3 & \textbf{17} & 2 & 5 & 4 & 2 & \textbf{29} & 2 & 3 & 8 & 2 & 8 & 3 & 2 & \textbf{46}
\\
\hline \rule[0pt]{0pt}{12pt}
\end{tabular}
\end{raggedright}\caption{\label{tab:fici2} The first few values of the maximum exponent $k_{m}$ of an abelian power of period $m$ in the Fibonacci word $f$. The values corresponding to the Fibonacci numbers are in bold.}
\end{small}
\end{table}

\subsection{Abelian critical exponent of Sturmian words}\label{sec:critStur}

Mignosi and Pirillo proved that the critical exponent of the Fibonacci word is $2+\phi$~\cite{MignosiPirillo}. In general, the critical exponent of a Sturmian word can be finite or infinite. The following theorem gives a characterization of Sturmian words with finite critical exponent.

\begin{theorem}\cite{Mi89,durand_2003}\label{thm:betapf}
 Let $s_{\alpha}$ be a Sturmian word of slope $\alpha$. The following are equivalent:
\begin{enumerate}
  \item $s_{\alpha}$ is $\beta$-free for some $\beta$;
 \item $\alpha$ has bounded partial quotients; 
 \item $s_{\alpha}$ is linearly recurrent. 
\end{enumerate}
\end{theorem}

Let $\alpha=[0;a_1,a_2,\ldots]$ and suppose that the sequence $(a_i)$ of partial quotients of $\alpha$ is bounded. Let $p_i/q_i=[0;a_1,a_2,\ldots,a_i]$ be the sequence of convergents of $\alpha$. Then the critical exponent $\chi(s_{\alpha})$ of $s_{\alpha}$ is given by (see~\cite{Carpi3,DBLP:journals/tcs/Vandeth00,DBLP:journals/ejc/DamanikL02})
\[
\chi(s_{\alpha})=\max \left\{ a_1, 2+\sup_{i\ge 2}\{a_i+(q_{i-1}-2)/q_i\}\right\}
\]
Thus, the critical exponent of the  Fibonacci word is the least critical exponent a Sturmian word can have.



Before studying the abelian critical exponent (see Definition~\ref{def:ace}) of Sturmian words further, we explore its connection to a number-theoretical concept known as
the Lagrange spectrum.

\begin{definition}
  Let $\alpha$ be a real number. The Lagrange constant of $\alpha$ is defined as
  \begin{equation*}
    \lambda(\alpha) = \limsup_{m\to\infty} (m\|m\alpha\|)^{-1}.
  \end{equation*}
\end{definition}

Let us briefly motivate the definition of the Lagrange constants. The famous Hurwitz's Theorem states that for every irrational $\alpha$ there exists infinitely many rational numbers $n/m$ such that $$\left|\alpha - \frac{n}{m}\right| < \frac{1}{\sqrt{5}m^2}$$ and, moreover, the constant $\sqrt{5}$ is best possible. Indeed, if $\alpha=\phi-1$, then for every $k>\sqrt{5}$ the inequality
$$\left|\frac{n}{m}-\alpha\right|< \frac{1}{km^2}$$
has only a finite number of solutions $n/m$.

For a general irrational $\alpha$, the infimum of the real numbers $\lambda$ such that for every $k>\lambda$ the inequality
$\left|n/m-\alpha\right|< 1/km^2$
has only a finite number of solutions $n/m$, is indeed the Lagrange constant $\lambda(\alpha)$ of $\alpha$. The set of all finite Lagrange
constants of irrationals is called the \emph{Lagrange spectrum} $L$. The Lagrange spectrum has been
extensively studied, yet its structure is still not completely understood. Markov  proved that
$L \cap (-\infty, 3) = \{\ell_1 = \sqrt{5} < \ell_2 = \sqrt{8} < \ell_3 = \sqrt{221}/5 < \ldots\}$
where $\ell_n$ is a sequence of quadratic irrational numbers  converging to $3$ (so the beginning of $L$ is discrete). Then Hall  proved  that $L$ contains a whole half line, and Freiman  determined the biggest half line that is contained
in $L$, which is $[c_F, +\infty)$, with
$$c_F=\frac{2221564096+283748\sqrt{462}}{491993569}=4.5278295661\ldots $$

Using the terminology of Lagrange constants, we have the following direct consequence of Theorem \ref{the:main1}.

\begin{theorem}\cite{tcs16}\label{thm:lag}
  Let $s_\alpha$ be a Sturmian word of slope $\alpha$. Then $\act(s_\alpha) = \lambda(\alpha)$. In other words, the
  abelian critical exponent of a Sturmian word is the Lagrange constant of its slope.
\end{theorem}

The abelian critical  exponent of the Fibonacci word is $\sqrt{5}$. It is the smallest possible. Indeed, from Theorem \ref{thm:lag} one gets the following result.

\begin{theorem}\cite{tcs16}\label{theor:sqrt5}
  For every Sturmian word $s_{\alpha}$ of slope $\alpha$, we have $\act(s_\alpha) \geq \sqrt{5}$. \end{theorem}

Actually, thanks to Theorem \ref{thm:lag}, one can obtain a formula to compute the  abelian critical exponent of a Sturmian word, as in the classical case:

\begin{proposition}\label{prp:act_formula}
  Let $s_\alpha$ be a Sturmian word of slope $\alpha$. Then the abelian critical exponent of $s_\alpha$ is
  \begin{equation*}
    \act(s_\alpha) = \limsup_{i\to +\infty} \left([a_{i+1};a_{i+2},\ldots]+[0;a_i,a_{i-1},\ldots,a_1]\right).
  \end{equation*}
\end{proposition}

In conclusion, one has the following generalization of Theorem \ref{thm:betapf} to the abelian case:

\begin{theorem}\cite{tcs16}\label{theor:act_finite}
  Let $s_\alpha$ be a Sturmian word of slope $\alpha$. The following are equivalent:
  \begin{enumerate}
    \item $\act(s_\alpha)$ is finite;
    \item $\alpha$ has bounded partial quotients;
 \item $s_\alpha$ is $\beta$-free for some $\beta$.
  \end{enumerate}
\end{theorem}

\subsection{Abelian periods of factors of Sturmian words}

The Fibonacci word has another remarkable property:  the smallest period of any of its finite factors is a Fibonacci number: 

\begin{proposition}\cite{DBLP:journals/ita/CurrieS09}
 The set of smallest periods of factors of the Fibonacci infinite word is the set of Fibonacci numbers.
\end{proposition}

This result can be generalized to abelian periods, in the sense of Definition~\ref{def:abper}.  For example, the smallest abelian period of $01001010=0\cdot 10 \cdot 01 \cdot 01 \cdot 0$ is $2$.

\begin{proposition}\cite{tcs16}\label{the:abperFib}
 The set of smallest abelian periods of factors of the Fibonacci infinite word is the set of Fibonacci numbers.\end{proposition}

For a general Sturmian word, Currie and Saari~\cite{DBLP:journals/ita/CurrieS09} characterized the set of the smallest periods of factors:

\begin{theorem}\cite{DBLP:journals/ita/CurrieS09}\label{the:csPer}
The set of smallest periods of factors of a Sturmian word of slope $\alpha$ having continued fraction expansion $[0;a_1,a_2,\ldots]$ is $\{\ell q_k+q_{k-1}\mid k\geq 0, \ell=1,2,\ldots,a_{k+1}\}$, where the sequence $(q_k)$ is the sequence of denominators of convergents of $\alpha$.
\end{theorem}

Peltom\"aki~\cite{PELTOMAKI2020251} gave a generalization of the latter result to the case of abelian periods, even though a full characterization seems more involved in this case:

\begin{theorem}\cite{PELTOMAKI2020251}\label{the:csPerAB}
If $m$ is the smallest abelian period of a nonempty factor of a Sturmian word of slope $\alpha$  having continued fraction expansion $[0;a_1,a_2,\ldots]$, then either $m=tq_k$ for some $k\geq 0$ and $1\leq t\leq a_{k+1}$ or $m=\ell q_k+q_{k-1}$ for some $k\geq 1$ and some $1\leq \ell\leq a_{k+1}$, where the sequence $(q_k)$ is the sequence of denominators of convergents of $\alpha$.
\end{theorem}

\subsection{Abelian returns}

\begin{definition}
 Let $x$ be an infinite recurrent word. A word $w$ is a \emph{first return} (or simply a \emph{return}) to a factor $u$ of $x$ if $wu$ is a factor of $x$ and $u$ occurs only twice in $wu$, as its prefix and as its suffix.
\end{definition}

In other words, given a factor $u$ of a recurrent word $x$, we know that $u$ must eventually reoccur  in $x$, and we consider the factors of $x$ between two consecutive occurrences of $u$ (which may overlap) in $w$. For example, in the Fibonacci word $f$, the returns to $101$ are $10100$ and $10100100$. 

\begin{theorem}\cite{DBLP:journals/ejc/Vuillon01}\label{thm:return}
 An infinite word is Sturmian if and only if each of its factors has exactly two returns.
\end{theorem}

This is once again tight, because if a factor of an infinite recurrent word $x$ has only one return, then $x$ is ultimately periodic. We now present an extension of this result
to the abelian case.

We consider two abelian modifications of the notion of return
word. Given a factor $u$ of an infinite word $x$, let $n_1 < n_2 <
n_3 < \ldots$ be all the integers $n_i$ such that $w_{n_i}\cdots
w_{n_{i+|u|-1}}$ is abelian equivalent to $u$. Then we call each
$w_{n_i}\cdots w_{n_{i+1}-1}$ a \emph{semi-abelian return} to
$u$. By an \emph{abelian return} to $u$ we mean the abelian class
of $w_{n_i}\cdots w_{n_{i+1}-1}$. We note that in both cases these
definitions depend only on the abelian class of $u$. For example, in the Fibonacci word, the word abelian class of $010$ has three abelian ands semi-abelian returns: $0$ (in the factors $0100$ and $0010$), $1$ (in the factor $1001$) and $01$ (in the factor $01010$).

Each of these notions of abelian returns gives rise to a
characterization of Sturmian words. Moreover, the
characterizations are the same in terms of abelian and
semi-abelian returns:

\begin{theorem}
 A binary recurrent infinite word $x$ is Sturmian if and only if each factor $u$ of $x$
has two or three (semi-)abelian returns in $x$.
\end{theorem}

In \cite{Rigo13}, the authors define the set $\mathcal{APR}_x$ as the set of all semi-abelian returns to all prefixes of an infinite word $x$. This definition gives a characterization of Sturmian words of intercept 0 among all other Sturmian words:

\begin{theorem}{\cite{Rigo13}}\label{RSV13:finite} Let $x$ be a Sturmian word. The set $\mathcal{APR}_x$ is finite if and only if $x$ does not
have a null intercept. \end{theorem}

In \cite{DBLP:conf/cwords/MasakovaP13}, the authors provide explicit formulas for the cardinality
of the set $\mathcal{APR}_x$  of abelian returns of all prefixes of a Sturmian
word $x$ in terms of the partial quotients of its slope,
depending on the intercept. They also provide a complete description of the set $\mathcal{APR}_x$   for characteristic Sturmian words. 

In \cite{DBLP:journals/tcs/RampersadRS14},  the result from Theorem \ref{RSV13:finite} is generalized to rotation words, another generalization of Sturmian words.
 Given $\alpha,\beta \in (0,1)$ and  $\rho \in [0,1)$, the \emph{rotation word } $r=r(\alpha,\beta,\rho)$ is the word $r=r_0r_1\cdots $ satisfying, for all $i\geq 0$,

$$r_i=\begin{cases} 1, & \mbox{if } R^i_{\alpha}(\rho)\in [1-\beta,0);\\ 0, & \mbox{otherwise}.\end{cases}$$ 

\begin{theorem}{\cite{DBLP:journals/tcs/RampersadRS14}} Let $\alpha$ be irrational. Let $m$ be an integer. Let $r=r(\alpha,\{m \alpha\},\rho)$ be a rotation word. The set $\mathcal{APR}_r$ is finite if and only if $\rho \notin \{\{-i\alpha \}|0\leq i<m\}$.\end{theorem}


\subsection{Minimal abelian squares}

A square is called minimal if it does not have square prefixes. For example, $0101$ is a minimal square, while $001001$ is not.

\begin{theorem}\label{thm:squares}~\cite{SAARI2010177}
Any aperiodic word contains at least $6$ minimal squares.
\end{theorem}

Every Sturmian word contains \emph{exactly} $6$ minimal squares --- however, there are aperiodic words with exactly $6$ minimal squares that are not Sturmian.

For example, the minimal squares of the Fibonacci word are: $00$, $0101$, $1010$, $010010$, $100100$ and $1001010010$. Moreover, in each position of the Fibonacci word  one of these squares starts.

Thus, one can also consider the decomposition of a Sturmian  word $s$ in these minimal squares. By deleting half of each square one obtains a new infinite word $\sqrt{s}$, and this word is again a Sturmian word and has the same slope of $s$~\cite{DBLP:journals/combinatorics/PeltomakiW17}. 

In the case of the Fibonacci word
\[f=010010 \cdot 100100 \cdot 1010 \cdot 0101 \cdot 00 \cdot 1001010010 \cdot 0101 \cdot 00 \cdot 1010 \cdots\]
one obtains the Sturmian word
\[
\sqrt{f}=010 \cdot 100 \cdot 10 \cdot 01 \cdot 0 \cdot 10010 \cdot 01 \cdot 0 \cdot 10 \cdots
\]

There is a generalization of Theorem \ref{thm:squares} to the abelian setting:

 \begin{theorem}~\cite{SAARI2010177}
Any aperiodic word contains at least $5$ minimal abelian squares.
\end{theorem}

Here, a minimal abelian square is one such that none of its proper prefixes is an abelian square.

Sturmian words have exactly $5$ minimal abelian squares, and in each position one of these $5$ minimal abelian squares  starts. For example, the minimal abelian squares of the Fibonacci word are $00$, $010010$, $0101$, $1001$ and $1010$.  

One can also consider the decomposition of a Sturmian  word $s$ in these minimal abelian squares. By deleting half of each abelian square one obtains a new infinite word $\sqrt[ab]{s}$, and this word is again a Sturmian word with the same slope of $s$~\cite{peltoPHD}. 

In the case of the Fibonacci word, for example, the decomposition in minimal abelian squares is
\[f= 010010\cdot 1001\cdot 00\cdot 1010\cdot 0101\cdot 00\cdot 1001\cdot 010010\cdot 0101\cdot 00\cdot 1010\]
and one has $\sqrt[ab]{f}=\sqrt{f}$.





\section{Modifications of abelian equivalence}\label{sec:modifications}

In these section we discuss some relevant modifications of the notion of abelian equivalence.

\subsection{$k$-abelian equivalence}

Let $k$ be a positive integer. Two words $u$ and $v$ are
\emph{$k$-abelian equivalent}, denoted by $u\sim_k v$, if $|u|_t =
|v|_t$ for every word $t$ of length at most $k$, where $|w|_t$ denotes the number of occurrences of the factor $t$ in $w$. This defines a
family of equivalence relations $\sim_k$, bridging the gap between
the usual notion of abelian equivalence (when $k=1$) and equality
(when $k=\infty$).

Equivalently, $u$ and $v$ are $k$-abelian equivalent if both the following conditions hold:
\begin{itemize}
\item $|u|_t = |v|_t$ for every word $t$ of length exactly $k$;
\item $\Pref_{k - 1}(u) = \Pref_{k - 1}(v)$ and
    $\Suff_{k - 1}(u) = \Suff_{k - 1}(v)$
    (or $u = v$, if $|u| < k - 1$ or $|v| < k - 1$).
\end{itemize}

For instance, $00101\sim_2 01001$, but $00101\nsim_2 00011$. It is
clear that $k$-abelian equivalence implies $k'$-abelian
equivalence for every $k' < k$. In particular, $k$-abelian equivalence for any $k\geq 2$ implies abelian
equivalence, that is, $1$-abelian equivalence.

\subsubsection{Avoidance}

Similarly to usual and abelian powers, we naturally define
a \emph{$k$-abelian $l$-power} as a concatenation of $l$ words that are
$k$-abelian equivalent one to another. The basic problem to consider is
$k$-abelian avoidability. We ask what is the size of the smallest
alphabet where $k$-abelian squares or cubes can be avoided, for a
fixed $k$. Clearly, the size of the smallest alphabet for
$k$-abelian avoidability lies between the smallest sizes of the
alphabet necessary for avoiding abelian and usual powers. For
example, as squares are avoidable over a 3-letter alphabet and abelian squares are avoidable over a 4-letter alphabet (see Theorem \ref{avoid_global})
, we have that the smallest alphabet over which
$k$-abelian squares are avoidable consists of 3 or 4
letters. 

\begin{theorem}\cite{DBLP:journals/tcs/HuovaKS12,DBLP:journals/tcs/Rao15}
The following holds:
\begin{itemize}
\item The longest ternary word which is 2-abelian square-free has
length 537, so there does not exist an infinite
2-abelian square-free word over a ternary alphabet.

\item  2-abelian-cubes are avoidable over a binary alphabet.

\item  3-abelian-squares are avoidable over a ternary alphabet.
\end{itemize}
\end{theorem}

Similarly to Theorem \ref{thm:avoidance} in the abelian case, the
following has been shown:

\begin{theorem}\cite{DBLP:journals/moc/RaoR16,DBLP:journals/siamdm/RaoR18} One can avoid 3-abelian-squares of period at least 3 in infinite
binary words, 2-abelian-squares of period at least 2 in
infinite ternary words, and 2-abelian squares of period more than 63 in infinite binary words. \end{theorem}

\subsubsection{Complexity}

Given an infinite word $x$, we consider the associated complexity
function $p_x^{(k)}=\left | ({\Fact(x) \cap \Sigma_d^{n}}) / {\sim_{k}}
\right |$, which counts the
number of $k$-abelian equivalence classes of factors of $x$ of
length $n$.

\begin{theorem}\cite{DBLP:journals/jct/KarhumakiSZ13}
Let $k$ be a positive integer and $x$ an aperiodic word. The following
conditions are equivalent:
\begin{itemize}
\item $x$ is Sturmian;
\item $p_x^{(k)}(n)=\begin{cases} n+1 & \mbox{ for } 0\leq n \leq 2k-1,
\\ 2k & \mbox{ for } n \geq 2k.\end{cases}$\end{itemize}
\end{theorem}

Interestingly, the $2$-abelian complexity of the Thue-Morse word is unbounded~\cite{DBLP:journals/actaC/KarhumakiSZ17} (unlike the abelian complexity). Moreover, the $2$-abelian complexity of the Thue-Morse word, as well as the period-doubling  word, is  a $2$-regular sequence \cite{DBLP:journals/combinatorics/ParreauRRV15}.

\begin{theorem}\cite{DBLP:journals/jct/KarhumakiSZ13} Fix $k\geq 1$. Let $x$ be an infinite word over a finite alphabet $\Sigma$ having bounded
$k$-abelian complexity. Let $D \subseteq \mathbb{N}$ be a set of
positive upper density, that is
$$\limsup_{n\to\infty}\frac{|D\cap \{1, 2,\ldots,n\}|}{n}
> 0.$$
Then, for every positive integer $N$, there exist $i$ and $l$ such
that $\{i,i+l,i+2l, \dots ,i+Nl\}\in D$ and the $N$ consecutive
blocks  $(x[i+jl, i+(j+1){l-1}])_{0\leq j\leq N-1}$ of length $l$ are pairwise $k$-abelian equivalent. In particular, $x$ contains
$k$-abelian powers for arbitrarily large $k$ .\end{theorem}

\subsubsection{$k$-abelian classes}

In this section, we deal with equivalence classes of $\Sigma_d^*$ under $k$-abelian equivalence.

\begin{theorem}\cite{DBLP:journals/jct/KarhumakiSZ13} Let $k\geq 1$ and $\Sigma_d$ a $d$-letter alphabet, $d\geq 2$. The number of $k$-abelian equivalence classes of
$\Sigma_d^n$ is $\Theta(n^{d^k - d^{k-1}})$. \end{theorem}

Now we describe rewriting rules of words, which preserve
$k$-abelian equivalence classes and give a characterization of
$k$-abelian equivalence.

Let $k \geq 1$ and let $u = u_1\cdots u_n$. Suppose that
there exist indices $i,j,l$ and $m$, with $i<j \leq l<m\leq
n-k+2$, such that $u{[i,i+k-1)} = u{[l,l+k-1)} = x\in\Sigma_d^{k-1}$
and $u{[j,j+k-1)} = u{[m,m+k-1)} = y\in\Sigma_d^{k-1}$. We thus have
$$u = u{[1,i)} \cdot u{[i,j)} \cdot u{[j,l)} \cdot u{[l,m)} \cdot u{[m..]},$$
where $u{[i..]}$ and $u{[l..]}$ begin with $x$ and $u{[j..]}$ and
$u{[m..]}$ begin with $y$. Note here that we allow $l = j$ (in
this case $y = x$). We define a \emph{$k$-switching} on $u$,
denoted by $S_{u,k}(i,j,l,m)$, as
\begin{equation}\label{eq:k-switching}
    S_{u,k}(i,j,l,m) = u[1,i) \cdot u[l,m) \cdot u[j,l) \cdot u[i,j) \cdot u[m..].
\end{equation}

Roughly speaking, the idea is to switch the positions of two
factors that both begin and end with the same factors of length
$k-1$, and we allow the situation where the factors can 
overlap.

\begin{example}
    Let $u = 0010101000101$ and let $v = S_{u,4}(2,3,4,11)$. By \eqref{eq:k-switching}, we have
    $v = 0 \cdot 0101000\cdot 1\cdot 0 \cdot 101$.
    One can check that $u \sim_4 v$.
\end{example}


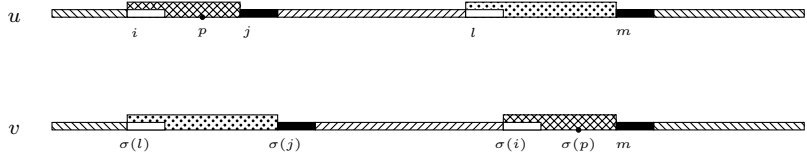
\begin{figure}
    \centering
        \begin{tikzpicture}
            \node at (-5.5,0) {\footnotesize{$u$}};
                \draw[pattern=north west lines] (-5,0) rectangle (-4,.1);
                \draw[pattern=north west lines] (3,0) rectangle (5,.1);
                \draw[pattern=north east lines] (-2,0) rectangle (.5,.1);
                \draw[pattern=crosshatch] (-4,0) rectangle (-2.5,.2);
                                \draw[fill=white] (-4,0) rectangle (-3.5,.1);
                    \node at (-3.9,-.2) {\tiny{$i$}};
                \draw[fill] (-3,0) circle (.03);
                    \node at (-3,-.2) {\tiny{$p$}};

                \draw[fill] (-2.5,0) rectangle (-2,.1);
                    \node at (-2.4,-.2) {\tiny{$j$}};

                \draw[fill] (2.5,0) rectangle (3,.1);
                \node at (2.6,-.2) {\tiny{$m$}};

                \draw[pattern=crosshatch dots] (.5,0) rectangle (2.5,.2);
                                \draw[fill=white] (.5,0) rectangle (1,.1);
                \node at (.6,-.2) {\tiny{$l$}};

            \node at (-5.5,-1.5) {\footnotesize{$v$}};
                \draw[pattern=north west lines] (-5,-1.5) rectangle (-4,-1.4);
                \draw[pattern=north west lines] (3,-1.5) rectangle (5,-1.4);
                \draw[pattern=north east lines] (-1.5,-1.5) rectangle (1,-1.4);
                \draw[pattern=crosshatch dots] (-4,-1.5) rectangle (-2,-1.3);
                                \draw[fill=white] (-4,-1.5) rectangle (-3.5,-1.4);
                    \node at (-3.9,-1.7) {\tiny{$\sigma(l)$}};
                \draw[fill] (-2,-1.5) rectangle (-1.5,-1.4);
                    \node at (-1.9,-1.7) {\tiny{$\sigma(j)$}};
                \draw[pattern=crosshatch] (1,-1.5) rectangle (2.5,-1.3);
                                \draw[fill=white] (1,-1.5) rectangle (1.5,-1.4);
                    \node at (1.1,-1.7) {\tiny{$\sigma(i)$}};
                    \draw[fill] (2,-1.5) circle (.03);
                    \node at (2,-1.7) {\tiny{$\sigma(p)$}};
                \draw[fill] (2.5,-1.5) rectangle (3,-1.4);
                    \node at (2.6,-1.7) {\tiny{$m$}};
        \end{tikzpicture}
        \caption{Illustration of a $k$-switching.}
        \label{fig:k-switching}
\end{figure}

Let us define a relation $R_k$ on words by $uR_kv$ if and only
if $u=v$ or $v = S_{u,k}$ for some $k$-switching of $u$. Now $R_k$
is clearly reflexive and symmetric. The transitive closure $R_k^*$
of $R_k$ is thus an equivalence relation. In fact, the relations
$\sim_k$ and $R_k^*$ actually coincide:


\begin{proposition}\cite{DBLP:journals/tcs/KarhumakiPRW17}\label{prop:rewritingCharacterization}
For two words $u,v$, we have $u\sim_k v$ if and only if $u
R_k^* v$.
\end{proposition}

This characterization can be used for studying the cardinality of
$k$-abelian equivalence classes. For example, using this
characterization, the following upper bound has been established
on the number of $k$-abelian \emph{singleton} classes, i.e.,
classes containing exactly one element:

\begin{theorem}\cite{DBLP:journals/tcs/KarhumakiPRW17}\label{th:singletons} The number of $k$-abelian singleton classes
is of order $O (n^{N_d(k-1)-1})$, where
$$N_d(l)=\tfrac{1}{l}\sum_{q\mid l}\varphi(q)d^{l/q}$$ is the number
of conjugacy classes of words of length $l$ over $\Sigma_d$ and $\varphi$
is the Euler's totient function. \end{theorem}

It is worth noticing that this bound is conjectured to be tight; in fact
 the number of $k$-abelian singleton classes
is of order $\Theta (n^{N_d(k-1)-1})$. In \cite{DBLP:journals/fuin/CassaigneKPW17}
it is proved that the sequences of the numbers of singletons, as well as the numbers of $k$-abelian classes of length $n$, are both $\mathbb{N}$-rational (see, e.g., \cite{DBLP:books/lib/BerstelR88} for the definition). Using this result, the following precise values for the numbers $S_{k,d}(n)$ of singular $k$-abelian classes of length $n$ over $\Sigma_d$ were obtained for small $k$ and small alphabets:

\begin{proposition}  \cite{DBLP:journals/fuin/CassaigneKPW17}

\begin{enumerate}
    \item For all $n \geq 4$, $S_{2,2}(n) = 2n + 4$;
    \item  For all $n \geq 9$, $S_{3,2}(n) = \frac{1}{2} n^2 + 16n +\frac{2}{3}\big(e^{\frac{2\pi i}{3}n}+e^{-\frac{2\pi i}{3}n} \big)-\frac{535}{12}-\frac{3}{4}(-1)^n$;
     \item  For all $n \geq 6$, $S_{2,3}(n) = 3n^2+27n-63$.
\end{enumerate}
\end{proposition}

Moreover, Whiteland~\cite[Proposition 6.7]{whitelandPhD} gave a formula for $S_{4,2}(n)$. 



Among other studies on $k$-abelian equivalence, we would like to
mention the classification of existence of $k$-abelian palindromic
poor words \cite{DBLP:journals/iandc/CassaigneKP18}, as well as a
$k$-abelian version of Fine and Wilf's Lemma~\cite{DBLP:journals/ijfcs/KarhumakiPS13}.



Finally, Peltom\"aki and Whiteland~\cite{PW20} extended the results of Sec.~\ref{sec:critStur} on the abelian critical exponent of Sturmian words to the case of $k$-abelian equivalence.

\subsection{Weak abelian equivalence}\label{subsec:WAP}

In this subsection we consider another modification of abelian
equivalence: Two finite words $u$ and $v$ are called \emph{weak abelian
equivalent} if they have the same frequencies of letters. For a
finite word $w\in\Sigma_d^+$, the \emph{frequency} $\rho_a(w)$ of a
letter $a\in\Sigma_d$ in $w$ is defined as
$\rho_a(w)=\frac{|w|_a}{|w|}$. In other words, in the case of weak
abelian equivalence only frequencies of letters are taken into
account, but not the lengths of the words. Clearly, weak abelian
equivalent words are abelian equivalent if and only if they have
the same lengths.

We define a \emph{weak abelian power} as a concatenation of weak
abelian equivalent words.
In \cite{pjm/1102784513} the authors explore the avoidance of weak abelian
powers:
\begin{theorem}
The following holds true:
\begin{itemize} \item Every binary word contains weak abelian $k$-powers for each $k$. \item There exists an infinite ternary word containing no weak abelian $(5^{11}+1)$-powers. \end{itemize}\end{theorem}
The number $(5^{11}+1)$ seems to be far from being 
optimal. 

Recall that an infinite word $w$ is called \emph{abelian 
periodic} if $w = v_0 v_1 \cdots$, where $v_k\in\Sigma_d^*$ for
$k\geq 1$, and $v_i \sim_{ab} v_j$ for all integers $i, j \geq 1$. 

\begin{definition} An infinite word $w$ is called weakly abelian periodic if $w = v_0 v_1 \cdots$, where $v_i\in
\Sigma_d^+$, $\rho_a(v_i) = \rho_a(v_j)$ for all $a\in\Sigma_d$ and
all integers $i, j\geq 1$.\end{definition}

In other words, a weakly abelian periodic word is an infinite
weakly abelian power (with a preperiod). 

\begin{definition} An infinite word $w$ is called bounded weakly abelian
periodic if it is weakly abelian periodic with bounded lengths of
blocks, i.e., there exists $C$ such that for every $i$ we have
$|v_i|\leq C$. \end{definition}

One can consider the following geometric interpretation of weak
abelian equivalence. Let $w=w_1 w_2 \cdots$ be a finite or
infinite word over a finite alphabet $\Sigma_d$. We translate
$w$ to a graph visiting points of the infinite rectangular grid by
interpreting letters of $w$ as drawing instructions. In the binary
case, we associate $0$ with a move by vector
$\textbf{v}_0=(1,-1)$, and $1$ with a move $\textbf{v}_1=(1,1)$.
We start at the origin $(x_0,y_0)=(0,0)$. At step $n$, we are at a
point $(x_{n-1}, y_{n-1})$ and we move by a vector corresponding
to the letter $w_{n}$, so that we come to a point $(x_{n},
y_{n})=(x_{n-1}, y_{n-1})+v_{w_n}$, and the two points $(x_{n-1},
y_{n-1})$ and $(x_{n}, y_{n})$ are connected with a line segment.
So, we translate the word $w$ to a path in $\mathbb{Z}^2$. We
denote the corresponding graph by $g_w$. Hence, the
graph of a word is a piecewise linear function with linear segments
connecting integer points (see Example 1). It is easy to see that
for weakly abelian equivalent words the final points of their
graphs and the origin are collinear, and weakly abelian periodic
word $w$ has a graph with infinitely many integer points on a line
with rational slope. 
Note that instead of the vectors $(1,-1)$ and $(1,1)$, one can use
any other pair of noncollinear vectors $\textbf{v}_0$ and
$\textbf{v}_1$. For a $k$-letter alphabet one can consider a
similar graph in $\mathbb{Z}^k$.

\begin{example} Recall the regular paperfolding word
$p=001001100011\cdots$  The
graph corresponding to the regular paperfolding word with
$\textbf{v}_0=(1,-1)$, $\textbf{v}_1=(1,1)$ is displayed in Fig.~\ref{fig_wap}. The
regular paperfolding word is not balanced and is weak abelian periodic
along the line $y=-1$ (and actually along any line $y=C$, $C=-1,
-2, \dots$). \end{example}

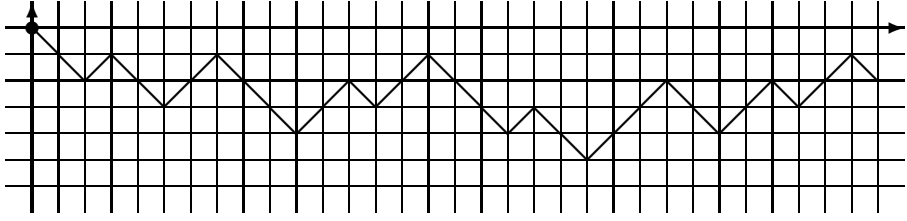
\begin{figure}
\centering
\begin{picture}(340,80)
\multiput(10,0)(10,0){33}{\line(0,1){80}}
\multiput(0,10)(0,10){7}{\line(1,0){340}}

\thicklines \put(0,70){\vector(1,0){340}}
\put(10,0){\vector(0,1){80}} \put(10,70){\circle*{5}}

\put(10,70){\line(1,-1){20}} \put(30,50){\line(1,1){10}}
\put(40,60){\line(1,-1){20}} \put(60,40){\line(1,1){20}}
\put(80,60){\line(1,-1){30}} \put(110,30){\line(1,1){20}}
\put(130,50){\line(1,-1){10}} \put(140,40){\line(1,1){20}}
\put(160,60){\line(1,-1){30}} \put(190,30){\line(1,1){10}}
\put(200,40){\line(1,-1){20}} \put(220,20){\line(1,1){30}}
\put(250,50){\line(1,-1){20}} \put(270,30){\line(1,1){20}}
\put(290,50){\line(1,-1){10}} \put(300,40){\line(1,1){20}}
\put(320,60){\line(1,-1){10}}

\end{picture}

        \caption{The
graph of the regular paperfolding word with $\textbf{v}_0=(1,-1)$,
$\textbf{v}_1=(1,1)$.}
        \label{fig_wap}
\end{figure}

In \cite{DBLP:journals/mst/AvgustinovichP16}, general properties of weak abelian periodicity are
studied; in particular, its relationships with the notions of balance and
letter frequency. Also, a characterization of weak abelian
periodicity of fixed points of binary uniform substitutions is
provided. 

Another result on weak abelian periodicity is a modification of a classical result of Ehrenfeucht and  Silberger on the relationship between periodicity and bordered factors, see Theorem \ref{th:borders}.
We say that a finite word $u$ is \emph{(weakly) abelian bordered} if $u$
contains a non-empty proper prefix which is (weakly) abelian
equivalent to a suffix of $u$.
Although Theorem \ref{th:borders} does not seem to generalize well for abelian equivalence relation (see a discussion in Subsection \ref{subsec:borders}), a similar assertion does hold, surprisingly, for weak
abelian periodicity:

\begin{theorem} \label{wap}\cite{DBLP:journals/mst/AvgustinovichP16}
    Let $w$ be an infinite word.
    If there exists a constant $C$ such that every factor $v$ of $w$ with $|v|\geq C$
    is weakly abelian bordered, then $w$ is bounded weakly abelian periodic.
    \end{theorem}

However, the converse of the previous statement does not hold. An example of a
weakly abelian periodic word with arbitrarily many weak abelian
unbordered factors is given by the word
$010^21^20^31^3\cdots0^n1^n\cdots$.

\subsection{$k$-binomial equivalence}

In this subsection we introduce another notion of equivalence
refining the abelian equivalence, namely, the $k$-binomial
equivalence.


The \emph{binomial coefficient}
$\left(\begin{array}{c}u\\v\end{array}\right)$ of two words $u$
and $v$ is defined as the number of occurrences of $v$ as a
scattered subword (see Sec.~\ref{subsec:borders} for the definition of scattered subword) in $u$.

The name comes from the fact that for two natural numbers $p>q$ and for a letter $a$, one has

$$\left(\begin{array}{c} a^p \\ a^q \end{array}\right) = \left(\begin{array}{c} p \\ q
\end{array}\right)$$

and

$$\left(\begin{array}{c} ua \\ vb \end{array}\right)  =
\begin{cases}
\left(\begin{array}{c} u \\ vb \end{array}\right) + \left(\begin{array}{c} u \\ v \end{array}\right), & \mbox { if }  a = b; \\
\left(\begin{array}{c} u \\ vb \end{array}\right), & \mbox {
otherwise. }
\end{cases}$$

\begin{definition}Two words $x$ and $y$ are $k$-\emph{binomially equivalent}, denoted by $x\sim_k^{bin} y$ if, for each word $v$ of length at most $k$, one has
$\left(\begin{array}{c}x\\v\end{array}\right) =
\left(\begin{array}{c}y\\v\end{array}\right)$. \end{definition}

In other words, two words are $k$-binomially equivalent if they
contain the same number of occurrences of subwords of length
at most $k$.

Since $ \left(\begin{array}{c}u\\a\end{array}\right) = |u|_a$ for
every $a \in \Sigma_d$, it is clear that two words $u$ and $v$ are
abelian equivalent if and only if $u \sim_1^{bin} v$. As it holds for
$k$-abelian equivalence, we have a family of refined relations: for
all $u, v \in \Sigma_d^*$, $k\geq 1$, $u \sim_{k+1}^{bin} v$ implies
$u \sim_{k}^{bin} v$.

\begin{example} The words $0101110$ and
$1001101$ are 2-binomially equivalent, since for both words we
have coefficients:  $ \left(\begin{array}{c}u\\0\end{array}\right)
= 3$,  $ \left(\begin{array}{c}u\\1\end{array}\right) = 4$,  $
\left(\begin{array}{c}u\\00\end{array}\right)  = 3$,  $
\left(\begin{array}{c}u\\01\end{array}\right) = 7$,  $
\left(\begin{array}{c}u\\10\end{array}\right) = 5$,  $
\left(\begin{array}{c}u\\11\end{array}\right) = 6$. On the other
hand, they are not 3-binomially equivalent:  As an example, we
have  $ \left(\begin{array}{c} 0101110 \\ 001 \end{array}\right)
= 3$ but $ \left(\begin{array}{c} 1001101 \\ 001
\end{array}\right) = 5$. Also, this example shows that the $k$-binomial equivalence is different from $k$-abelian equivalence:
these two words are clearly not $2$-abelian equivalent.
\end{example}

The following proposition gives the growth order of the number of
$m$-binomial classes for binary words:

\begin{proposition} \cite{DBLP:journals/tcs/RigoS15,1930-5346_2021032} Let $k \geq 2$. We have


$$\Sigma_2^n/\sim_k^{bin} \in O (n^{(k-1)2^{k-1}+1} ). $$ 

In particular, for $k=2$,

$$\Sigma_2^n/\sim_2^{bin} = n^3 + 5n + 6.$$ \end{proposition}

This bound can be extended to non-binary alphabets:
\begin{proposition} \cite{DBLP:journals/ijac/LejeuneRR20}
 Let $k\geq 1$. We have
$$ \Sigma_m^n/\sim_k^{bin} \in O (n^{k^2}m^{k}).$$
\end{proposition}

\begin{definition}The $k$-\emph{binomial complexity} $b_x^{(k)}$ of an infinite word $x$ over $\Sigma_d$ maps an integer
$n$ to the number of $k$-binomial equivalence classes of factors
of length $n$ occurring in $x$:
\[b_x^{(k)}(n) =|(\Fact(x)  \cap  \Sigma_d^n)/\sim_k^{bin}|.\]
\end{definition}

Note that $b_x^{(1)}$ corresponds to the usual abelian complexity
$a_x$. We have the following relations: for all $k\geq 1$, $b^{(k)}_x (n) \leq
b^{(k+1)}_x (n)$ and $a_x (n) \leq b^{(k)}_x (n) \leq p_x(n)$.

\begin{theorem}\label{thm:bin_st}\cite{DBLP:journals/tcs/RigoS15}
Let $k \geq 2$. If $x$ is a Sturmian word, then $b^{(k)}_x (n) = n
+ 1$ for all $n \geq 0$. \end{theorem}

\begin{remark} If $x$ is a right-infinite word such that $b^{(1)}_x (n)
= 2$ for all $n \geq 0$, then $x$ is clearly balanced. If
$b^{(2)}_x (n) = n + 1$, for all $n \geq 0$, then the factor
complexity function $p_x$ is unbounded and $x$ is aperiodic. As a
consequence of Theorem \ref{thm:bin_st}, an infinite word $x$ is
Sturmian if and only if, for all $n \geq 0$ and all $k \geq 2$,
$b^{(1)}_x (n) = 2$ and $b^{(k)}_x (n) = n + 1$.\end{remark}

A similar result holds for the Tribonacci word, which belongs to the family of Arnoux--Rauzy words (see the Preliminaries section for the formal definition of Arnoux--Rauzy words).
The factor complexity of every ternary Arnoux--Rauzy word is equal to $2n+1$, and for the Tribonacci word it turns out to be equal to its $k$-binomial complexity:
\begin{theorem} \cite{DBLP:journals/aam/LejeuneRR20}
Let $k \geq 2$ and $\tr$ be the Tribonacci word. Then $b^{(k)}_{\tr} (n) = 2n
+ 1$ for all $n \geq 0$.
\end{theorem}
The proof is surprisingly involved and is completely different from the proof for Sturmian words.

In contrast with Sturmian words and the Tribonacci word, which have the same binomial
 and factor complexity, certain morphic words (in
particular, the Thue--Morse word) have a bounded $k$-binomial
complexity.

\begin{definition} Let $\varphi$ be a substitution. If $\varphi(a)
\sim_{ab} \varphi(b)$ for all $a, b \in \Sigma_d$, then $\varphi$ is
said to be Parikh-constant. In particular, a
Parikh-constant substitution is $m$-uniform for some $m$, i.e., for
all $a \in \Sigma_d$, $|\varphi(a)| = m$. \end{definition}

\begin{theorem}\cite{DBLP:journals/tcs/RigoS15} Let $x$ be an infinite word that is a fixed point of a
Parikh-constant substitution. Let $k \geq 2$. There exists a constant
$C > 0$ (depending on $x$ and $k$) such that the $k$-binomial
complexity of $x$ satisfies $b^{(k)}_x (n) \leq C$ for all $n \geq
0$. \end{theorem}

This result has recently been extended to Parikh-collinear substitutions, i.e., substitutions such that the images of all letters have collinear Parikh vectors
\cite{DBLP:conf/dlt/RigoSW22}. Equivalently, Parikh-collinear morphisms can be defined as morphisms which map all infinite
words to words with bounded abelian complexity~\cite{cassaigne2011avoiding}.

Similarly to other modifications of abelian equivalence, we can
define a $k$-binomial square (resp., cube or $l$-power) as a
concatenation of two (resp., three or $l$) $k$-binomial equivalent
words. The natural questions concern avoidability and minimal
sizes of the alphabets which allow one to avoid certain powers.

\begin{theorem}\cite{DBLP:journals/tcs/RaoRS15}
2-binomial squares (resp. cubes) are avoidable over a 3-letter
(resp. 2-letter) alphabet. The sizes of the alphabets are optimal.
\end{theorem}

\begin{remark} An example of an infinite word avoiding
$2$-binomial squares (resp., cubes) is given by the fixed point $x =
012021012102012021020121\cdots$ (resp., $y =
001001011001001011001011011 \cdots$) of the substitution $g$ (resp.,
$h$):

$$g : \begin{cases}
 0 \mapsto 012, \\ 1 \mapsto 02, \\ 2 \mapsto 1; \end{cases} \qquad h : \begin{cases}
0 \mapsto 001, \\ 1 \mapsto 011.
\end{cases}$$

\end{remark}






\subsection{Additive powers}

In this subsection we assume our finite alphabet is a subset
of $\mathbb{N}$. An \emph{additive $k$-power} is a finite
nonempty word of the form $x_1x_2 \cdots x_k$ where $|x_1| = \cdots
= |x_k|$ and $\sum x_1 = \sum x_2 = \cdots = \sum x_k$, where by
$\sum x_i$ we mean the sum of the elements appearing in the word
$x_i$. It is worth mentioning that a modification of this definition without the condition on equal lengths is less interesting, since additive $k$-powers
 are unavoidable if the words do not have to have the same length~\cite{HH00}. Since two words of the same length over $\{0, 1\}$ have the
same sum if and only if they are permutations one of each other,
Dekking's result on avoiding abelian $4$-powers in binary words (see~Theorem~\ref{avoid_global})
shows that it is possible to avoid additive $4$-powers.

\begin{theorem}{\cite{DBLP:journals/jacm/CassaigneCSS14}} The fixed point $$x= 031430110343430310110110314303434303434303143011031011011031011 \cdots$$ of the substitution
$ 0 \to 03, 1 \to 43, 3 \to 1, 4 \to 01$ avoids additive
cubes.\end{theorem}

Moreover, additive cubes can be avoided for any alphabet which is not equivalent to $\{0,1,2,3\}$ in the following sense (the remaining case of $\{0,1,2,3\}$  is open so far):

\begin{theorem} [\cite{DBLP:conf/dlt/LietardR20, Lietard}] For any set $\Sigma \subseteq \mathbb{N}$ of size $4$ such that $\Sigma$ cannot be obtained by applying the same affine function to all the elements of $\{0, 1, 2, 3\}$, there is an infinite word over $\Sigma$ avoiding additive sums. \end{theorem}

However, the size of the alphabet to avoid additive cubes considered in the previous theorems is not
optimal. Since an abelian cube is necessarily an additive cube,
and we know it is impossible to avoid abelian cubes over an
alphabet of size 2, the alphabet size cannot be 2. The minimal
size of the alphabet to avoid additive cubes has recently been
shown to be $3$ by M. Rao \cite{DBLP:journals/tcs/Rao15}. The question on
whether it is possible to avoid additive squares remains open.
However, it is possible to avoid additive squares over $\mathbb{Z}^2$ (with componentwise addition defined on vectors):

\begin{theorem}\cite{DBLP:journals/siamdm/RaoR18} The fixed point
$h_{add}^{\infty}\left(\begin{array}{c}0\\0\end{array}\right)$ of
the following substitution does not contain any additive square.

$$h_{add} : \begin{cases} \left(\begin{array}{c}0\\0\end{array}\right)\to
\left(\begin{array}{c}0\\0\end{array}\right)\left(\begin{array}{c}2\\1\end{array}\right)\left(\begin{array}{c}2\\0\end{array}\right)
\qquad \left(\begin{array}{c}1\\1\end{array}\right)\to
\left(\begin{array}{c}0\\0\end{array}\right)\left(\begin{array}{c}0\\1\end{array}\right)\left(\begin{array}{c}1\\0\end{array}\right)\\
\left(\begin{array}{c}2\\1\end{array}\right)\to
\left(\begin{array}{c}1\\1\end{array}\right)\left(\begin{array}{c}0\\1\end{array}\right)\left(\begin{array}{c}1\\0\end{array}\right)
\qquad \left(\begin{array}{c}0\\1\end{array}\right)\to
\left(\begin{array}{c}1\\1\end{array}\right)\left(\begin{array}{c}0\\1\end{array}\right)\left(\begin{array}{c}2\\1\end{array}\right)\\
\left(\begin{array}{c}2\\0\end{array}\right)\to
\left(\begin{array}{c}0\\0\end{array}\right)\left(\begin{array}{c}1\\0\end{array}\right)\left(\begin{array}{c}2\\0\end{array}\right)
\qquad \left(\begin{array}{c}1\\0\end{array}\right)\to
\left(\begin{array}{c}1\\1\end{array}\right)\left(\begin{array}{c}2\\1\end{array}\right)\left(\begin{array}{c}2\\0\end{array}\right)\\\end{cases}$$
\end{theorem}

\section{Miscellanea}\label{sec:misc}

\subsection{Abelian subshifts}\label{subsec:ab_subshifts}

In the subsection, we consider an abelian version of the symbolic dynamical notion of a subshift. 
The subsection is based on \cite{DBLP:journals/jcta/PuzyninaW22,DBLP:journals/corr/abs-2012-14701,DBLP:conf/dlt/KarhumakiPW18}.





Similarly to the notion of a subshift (see Preliminaries section), the abelian subshifts are defined as follows: For a subshift
$X\subseteq\Sigma^{\mathbb{N}}$ the \emph{abelian subshift} of
$X$ is defined as
\begin{equation*}
\mathcal A_{X}= \{ y\in \Sigma^{\mathbb{N}}\colon \forall u\in \Fact( y), \exists v\in \Fact(X) \mbox{ with }u\sim_{\text{ab}}v\}.
\end{equation*}
Taking as $X$ a subshift generated by an infinite word $x$, one has an abelian subshift $A_x$ generated by the infinite word $x$. Observe that for any $ x\in \Sigma^{\mathbb{N}}$ the abelian subshift
$\mathcal A_{ x}$ is indeed a subshift. 

\begin{example}[Thue--Morse word] \label{ex:TM}
Consider the abelian subshift of the Thue-Morse word $\tm$.
For odd lengths $\tm$ has two abelian factors, and for even
lengths three. Further, the number of occurrences of $1$ in each
factor differs by at most $1$ from half of its length. It is easy
to see that any factor of any word in $\{\varepsilon,0,1\}\cdot
\{01,10\}^{\mathbb{N}}$ has the same property, i.e., $\{\varepsilon,0,1\}\cdot
\{01,10\}^{\mathbb{N}}\subseteq \mathcal A_{\tm}$. In fact, equality
holds: $\mathcal A_{\tm}=\{\varepsilon,0,1\}\cdot \{01,10\}^{\mathbb{N}}.$
Indeed, let $ x\in \mathcal A_{\tm}$. Then $x$ has
blocks of each letter of length at most 2 (since there are no
factors $000$ and $111$). Moreover, between two consecutive
occurrences of $00$ there must occur $11$, and vice versa
(otherwise we have a factor $001010\cdots 0100$, where the number
of occurrences of $1$ differs by more than $1$ from half of its
length). Clearly, such word is in $\{\varepsilon,0,1\}\cdot
\{01,10\}^{\mathbb{N}}.$ So, for the Thue-Morse word, its subshift is huge
compared to $\Omega_{\tm}$: basically, it is a morphic image of
the full binary shift.
\end{example}


\subsubsection{On abelian subshifts of binary
words}\label{sec:binaryWords} 

The following theorem gives a characterization of Sturmian words among binary words,
in terms of abelian subshifts. We remark that purely
periodic balanced words are sometimes also called Sturmian (i.e., one can consider Sturmian words with rational slope), and we follow this terminology in this section.
\begin{theorem}\label{th:St}\cite{DBLP:conf/dlt/KarhumakiPW18}
Let $ x\in\{0,1\}^{\mathbb{N}}$ be a uniformly recurrent aperiodic word.
Then $\mathcal A_{ x}$ contains exactly one minimal subshift if
and only if $x$ is Sturmian.
\end{theorem}

An equivalent statement of the previous theorem is the following: 

\begin{theorem}\label{th:St1}Let $x$ be an aperiodic binary word. Then  $\mathcal A_{ x}=\Omega_x$ if and only if $x$ is Sturmian. \end{theorem}

However, none of the characterizations extends to non-binary
alphabets:
Let $ f = 010010100\dots$ be the Fibonacci word and let $\varphi:
0\mapsto 02, 1\mapsto 12$. Then for $ w= \varphi ( f)$ one has
$\mathcal A_{ w} = \Omega_{ w}$ (see
\autoref{thm:AbelianSSMinimalComplexity}).

The following theorem shows that abelian subshifts of non-Sturmian uniformly recurrent binary words cannot contain finitely many minimal subshifts:

\begin{theorem}\cite{DBLP:journals/jcta/PuzyninaW22}
Let $x$ be a binary uniformly recurrent word which is not
aperiodic or periodic Sturmian. Then $\mathcal A_{ x}$ contains
infinitely many minimal subshifts.
\end{theorem}


This theorem, however, does not extend to the non-binary case either (see the next subsection).

\subsubsection{On abelian subshifts of generalizations of Sturmian words to nonbinary alphabets}\label{sec:minimalComplexity}

A natural question is how  the characterization of $\Omega_x=\mathcal A_{ x}$ from Theorem
\ref{th:St1} extends to nonbinary alphabets. The problem is open so far: 

\begin{problem}
Characterize aperiodic non-binary words $x$ such that  $\Omega_x=\mathcal A_{ x}$.\end{problem}

In this section, we will see that the property $\Omega_x=\mathcal A_{ x}$ does not characterize natural generalizations of Sturmian words to nonbinary alphabets following different equivalent definitions of Sturmian words; in particular, words of minimal
complexity, balanced
 words and Arnoux--Rauzy words.

We start with aperiodic nonbinary words
of minimal complexity. Over an alphabet $\Sigma_d$, the minimal
factor complexity of an aperiodic word is known to be $n+d-1$. The structure of words of complexity
$n+C$ is related to the structure of Sturmian words and is well
understood (see~\cite{FM97,Didier99,KaboreTapsoba07}). 
We will make use of the following description of aperiodic ternary words of minimal complexity:



\begin{theorem}[\cite{FM97} as formulated in \cite{KaboreTapsoba07}]\label{thm:characterizationNp2Complexity}
A recurrent infinite word $x$ over $\Sigma_3$ has factor complexity $p_{x}(n)
= n+2$ for all $n\geq 1$ if and only if 
(up to permuting the letters)
$ x \in \Omega_{\varphi( s)}$, where $ s$ is a
Sturmian word over $\Sigma_2$ and $\varphi$ is defined
\begin{enumerate}[resume,topsep=0pt]
\item  either by \label{item:item2LowComplexity} $0\mapsto 02$, $1\mapsto 12$;
\item or by \label{item:item3LowComplexity} $0\mapsto 0$, $1\mapsto 12$.
\end{enumerate}
\end{theorem}

The abelian subshifts behave differently for a ternary alphabet and for larger alphabets: 


\begin{theorem}\label{thm:AbelianSSMinimalComplexity} \cite{DBLP:journals/corr/abs-2012-14701}
Let $x$ be a recurrent word of factor complexity $n+C$ for all $n\geq
1$. 

\begin{enumerate} \item For $C=2$, if $ x$ is as in
\autoref{thm:characterizationNp2Complexity},
\autoref{item:item2LowComplexity}, then $\mathcal A_{ x} =
\Omega_{x}$. If $ x$ is as in
\autoref{item:item3LowComplexity}, then $\mathcal A_{ x}$
contains uncountably many minimal subshifts.

\item If $C>2$, then $\mathcal A_{x}$ contains exactly two
minimal subshifts.

\end{enumerate}
\end{theorem}




Now we examine aperiodic uniformly recurrent balanced words and
their abelian subshifts. The structure of such words is well understood and it has been described in  \cite{DBLP:journals/jct/Graham73,Hubert00}. 

\begin{theorem}\label{thm:BalancedAbelianSS} \cite{DBLP:journals/corr/abs-2012-14701}
Let $ x$ be aperiodic recurrent and balanced. Then $\mathcal
A_{ x}$ is the union of finitely many minimal subshifts.
\end{theorem}

Depending on the balanced word, its abelian subshift can contain either one or several minimal subshifts. In fact, for any integer $k$, there exist words (not necessarily balanced) with abelian subshifts containing exactly $k$ minimal subshifts \cite{DBLP:journals/corr/abs-2012-14701}.




\smallskip

Finally, we discuss abelian subshifts of Arnoux--Rauzy words~\cite{ArRa,DJP}. Apparently, the structure of abelian subshifts of Arnoux--Rauzy
words is rather complicated. For example, it is not hard to see
that for any Arnoux--Rauzy word with a characteristic word $c$ its
abelian subshift contains $20c$ (here we assume that 0 is the
first letter of the directive word $x$ and $2$ is the third letter occurring in
$x$ for the first time, i.e., $x$ has a prefix of the
form $0\{0,1\}^*1 \{0,1\}^* 2$). On the other hand, $20c \notin
\Omega_c$, so ${\mathcal{A}}_w\neq \Omega_w$ for an Arnoux--Rauzy word $w$. Hejda, Steiner and Zamboni studied the abelian shift
of the Tribonacci word $\tr$. They announced that $\mathcal
A_{\tr}\setminus \Omega_{\tr}\neq\emptyset$ but that
$\Omega_{\tr}$ is the only minimal subshift contained in
$\mathcal A_{\tr}$
\cite{HejdaSteinerZamboni15,ZamboniPersonal}.

An interesting open question is to understand the general
structure of the abelian subshifts of Arnoux--Rauzy words:

\begin{problem}
Characterize abelian subshifts of Arnoux--Rauzy words.\end{problem}

\subsection{Rich words and abelian equivalence}

In this subsection, we exhibit a nice fact relating palindromes and abelian equivalence. It is easy to see that a finite word of length $n$ contains at most $n+1$ distinct palindromes (including the empty word). Indeed,  adding a letter to a word, one can introduce at most one new palindrome. Words of length $n$ containing $n+1$ distinct palindromes are called \emph{rich}.

\begin{proposition}\cite{DBLP:journals/ejc/GlenJWZ09}
Any two rich words with the same set of palindromic factors are abelian equivalent.
\end{proposition}

\begin{proof} Let $w$ and $w'$ be two distinct rich words with the same set of palindromic factors. Any palindromic factor of $w$ (resp. $w'$) ending (and hence beginning) with a letter $x \in \Sigma$ is the
unique palindromic suffix of some prefix of $w$ (resp. $w'$). Thus the number of $x$'s in $w$ (resp. $w'$) is the number of palindromic
factors ending with $x$. So, $ |w|_x =| w'|_x$ for each letter $x \in \Sigma$. Therefore, $w$ and $w'$
are abelian equivalent.  \end{proof}

The converse of the previous proposition does not hold true, in general. For example, $001$ and $010$ are abelian equivalent rich words but their palindromic factors are different.

\subsection{Abelian saturated words}

Let $f$ be an increasing function. An infinite word $w$ is called abelian $f(n)$-\emph{saturated} if there exists a constant $C$ such that each factor of length $n$ contains at least $C f(n)$ abelian nonequivalent factors. 

\begin{theorem}\cite{DBLP:journals/combinatorics/Puzynina19} A binary infinite word cannot be abelian $n^2$-saturated, but, for any $\varepsilon>0$, there exist abelian $n^{2-\varepsilon}$-saturated binary infinite words. 
\end{theorem}

The examples of such words can be built using uniform morphisms of the following form:

$$\sigma: a\mapsto a^Kb, b\mapsto ab^K.$$
Choosing $K$ large enough, one gets $n^{2-\varepsilon}$-saturated binary infinite words \cite{DBLP:journals/combinatorics/Puzynina19}. The existence of abelian $n^2$-saturated infinite words over larger alphabets is an open question:

\begin{problem}
Do there exist abelian $n^2$-saturated infinite words over alphabets of cardinality more than 2?
\end{problem}

\section{Acknowledgments}

We thank the anonymous reviewers for their careful reading and helpful comments. We also thank James Currie, Jarkko Peltom\"aki, Narad Rampersad, Michel Rigo, Markus Whiteland and Luca Zamboni for reading a preliminary version of this paper and providing many valuable suggestions.

\bibliographystyle{abbrv}

\begin{thebibliography}{100}

\bibitem{ACR04}
A.~Aberkane, J.~Currie, and N.~Rampersad.
\newblock The number of ternary words avoiding abelian cubes grows
  exponentially.
\newblock {\em J. Integer Seq.}, 7:04.2.7, 2004.

\bibitem{DBLP:journals/tcs/Adamczewski03}
B.~Adamczewski.
\newblock Balances for fixed points of primitive substitutions.
\newblock {\em Theoret. Comput. Sci.}, 307(1):47--75, 2003.

\bibitem{DBLP:journals/dam/AgoB21}
K.~Ago and B.~Basic.
\newblock On highly palindromic words: The $n$-ary case.
\newblock {\em Discret. Appl. Math.}, 304:98--109, 2021.

\bibitem{AS03}
J.-P. Allouche and J.~O. Shallit.
\newblock {\em Automatic Sequences -- Theory, Applications, Generalizations}.
\newblock Cambridge University Press, 2003.

\bibitem{ArRa}
P.~Arnoux and G.~Rauzy.
\newblock Repr\'esentation g\'eom\'etrique de suites de complexit\'e $2n+1$.
\newblock {\em Bulletin de la Soci\'{e}t\'{e} Math\'{e}matique de France},
  119:199--215, 1991.

\bibitem{AKP2012}
S.~Avgustinovich, J.~Karhum{\"a}ki, and S.~Puzynina.
\newblock On abelian versions of {C}ritical {F}actorization {T}heorem.
\newblock {\em RAIRO Theor. Inform. Appl.}, 46:3--15, 2012.

\bibitem{DBLP:journals/jalc/AvgustinovichF02}
S.~V. Avgustinovich and A.~E. Frid.
\newblock Words avoiding abelian inclusions.
\newblock {\em J. Autom. Lang. Comb.}, 7(1):3--9, 2002.

\bibitem{DBLP:journals/mst/AvgustinovichP16}
S.~V. Avgustinovich and S.~Puzynina.
\newblock Weak abelian periodicity of infinite words.
\newblock {\em Theory Comput. Syst.}, 59(2):161--179, 2016.

\bibitem{DBLP:books/lib/BerstelR88}
J.~Berstel and C.~Reutenauer.
\newblock {\em Rational series and their languages}, volume~12 of {\em {EATCS}
  monographs on theoretical computer science}.
\newblock Springer, 1988.

\bibitem{berthe_rigo_CANT}
V.~Berth\'e and M.~Rigo, editors.
\newblock {\em Combinatorics, Automata and Number Theory}, volume 135 of {\em
  Encyclopedia of Mathematics and its Applications}.
\newblock Cambridge University Press, 2010.

\bibitem{Sadri22}
F.~Blanchet-Sadri, K.~Chen, and K.~Hawes.
\newblock Dyck words, lattice paths, and abelian borders.
\newblock {\em Internat. J. Found. Comput. Sci.}, 33:203--226, 2022.

\bibitem{DBLP:journals/int/Blanchet-SadriC14}
F.~Blanchet{-}Sadri, J.~Currie, N.~Rampersad, and N.~Fox.
\newblock Abelian complexity of fixed point of morphism 0 {\(\mapsto\)} 012, 1
  {\(\mapsto\)} 02, 2 {\(\mapsto\)} 1.
\newblock {\em Integers}, 14:A11, 2014.

\bibitem{DBLP:journals/aam/Blanchet-SadriF14}
F.~Blanchet{-}Sadri, N.~Fox, and N.~Rampersad.
\newblock On the asymptotic abelian complexity of morphic words.
\newblock {\em Adv. Appl. Math.}, 61:46--84, 2014.

\bibitem{DBLP:journals/tcs/Blanchet-SadriS16}
F.~Blanchet{-}Sadri, D.~Seita, and D.~Wise.
\newblock Computing abelian complexity of binary uniform morphic words.
\newblock {\em Theoret. Comput. Sci.}, 640:41--51, 2016.

\bibitem{BL22}
S.~Brlek and S.~Li.
\newblock On the number of squares in a finite word.
\newblock {\em CoRR}, abs/2204.10204, 2022.

\bibitem{DBLP:journals/ijac/Carpi93}
A.~Carpi.
\newblock On abelian power-free morphisms.
\newblock {\em Internat. J. Algebra Comput.}, 3(2):151--168, 1993.

\bibitem{DBLP:journals/dam/Carpi98}
A.~Carpi.
\newblock On the number of abelian square-free words on four letters.
\newblock {\em Discret. Appl. Math.}, 81(1-3):155--167, 1998.

\bibitem{Carpi3}
A.~Carpi and A.~de~Luca.
\newblock Special factors, periodicity, and an application to {S}turmian words.
\newblock {\em Acta Inf.}, 36:983--1006, 2000.

\bibitem{DBLP:journals/ejc/CassaigneC99}
J.~Cassaigne and J.~Currie.
\newblock Words strongly avoiding fractional powers.
\newblock {\em European J. Combin.}, 20(8):725--737, 1999.

\bibitem{DBLP:journals/jacm/CassaigneCSS14}
J.~Cassaigne, J.~Currie, L.~Schaeffer, and J.~O. Shallit.
\newblock Avoiding three consecutive blocks of the same size and same sum.
\newblock {\em J. {ACM}}, 61(2):10:1--10:17, 2014.

\bibitem{AIF_2000__50_4_1265_0}
J.~Cassaigne, S.~Ferenczi, and L.~Q. Zamboni.
\newblock Imbalances in {Arnoux-Rauzy} sequences.
\newblock {\em Annales de l'Institut Fourier}, 50(4):1265--1276, 2000.

\bibitem{DBLP:journals/iandc/CassaigneKP18}
J.~Cassaigne, J.~Karhum{\"{a}}ki, and S.~Puzynina.
\newblock On \emph{k}-abelian palindromes.
\newblock {\em Inf. Comput.}, 260:89--98, 2018.

\bibitem{DBLP:journals/fuin/CassaigneKPW17}
J.~Cassaigne, J.~Karhum{\"{a}}ki, S.~Puzynina, and M.~A. Whiteland.
\newblock $k$-abelian equivalence and rationality.
\newblock {\em Fund. Inform.}, 154(1-4):65--94, 2017.

\bibitem{cassaigne2011avoiding}
J.~Cassaigne, G.~Richomme, K.~Saari, and L.~Zamboni.
\newblock Avoiding abelian powers in binary words with bounded abelian
  complexity.
\newblock {\em Internat. J. Found. Comput. Sci.}, 22(4):905--920, 2011.

\bibitem{CV78}
Y.~C\'esari and M.~Vincent.
\newblock Une caract\'erisation des mots p\'eriodiques.
\newblock {\em C.R. Acad. Sci. Paris}, 286(A):1175--1177, 1978.

\bibitem{DBLP:journals/ejc/CharlierHPZ16}
{\'{E}}.~Charlier, T.~Harju, S.~Puzynina, and L.~Q. Zamboni.
\newblock Abelian bordered factors and periodicity.
\newblock {\em European J. Combin.}, 51:407--418, 2016.

\bibitem{DBLP:journals/jct/CharlierKPZ14}
{\'E}.~Charlier, T.~Kamae, S.~Puzynina, and L.~Q. Zamboni.
\newblock Infinite self-shuffling words.
\newblock {\em J. Combin. Theory Ser. {A}}, 128:1--40, 2014.

\bibitem{DBLP:reference/hfl/ChoffrutK97}
C.~Choffrut and J.~Karhum{\"{a}}ki.
\newblock Combinatorics of words.
\newblock In G.~Rozenberg and A.~Salomaa, editors, {\em Handbook of Formal
  Languages, Volume 1: Word, Language, Grammar}, pages 329--438. Springer,
  1997.

\bibitem{1930-5346_2021032}
J.~Chrisnata, H.~M. Kiah, S.~R. Karingula, A.~Vardy, E.~Y. Yao, and H.~Yao.
\newblock On the number of distinct $k$-decks: Enumeration and bounds.
\newblock {\em Advances in Mathematics of Communications}, 2022.

\bibitem{DBLP:journals/dam/ChristodoulakisCCI14}
M.~Christodoulakis, M.~Christou, M.~Crochemore, and C.~S. Iliopoulos.
\newblock Abelian borders in binary words.
\newblock {\em Discret. Appl. Math.}, 171:141--146, 2014.

\bibitem{Ch14}
M.~Christodoulakis, M.~Christou, M.~Crochemore, and C.~S. Iliopoulos.
\newblock On the average number of regularities in a word.
\newblock {\em Theoret. Comput. Sci.}, 525:3--9, 2014.

\bibitem{CI2006}
S.~Constantinescu and L.~Ilie.
\newblock {F}ine and {W}ilf's theorem for abelian periods.
\newblock {\em Bull. Eur. Assoc. Theoret. Comput. Sci. EATCS}, 89:167--170,
  2006.

\bibitem{Coven-Hedlund}
E.~Coven and G.~Hedlund.
\newblock Sequences with minimal block growth.
\newblock {\em Math. Systems Theory}, 7:138--153, 1973.

\bibitem{DBLP:journals/tcs/Currie04}
J.~Currie.
\newblock The number of binary words avoiding abelian fourth powers grows
  exponentially.
\newblock {\em Theoret. Comput. Sci.}, 319(1-3):441--446, 2004.

\bibitem{CurrieLinek2001}
J.~Currie and V.~Linek.
\newblock Avoiding patterns in the abelian sense.
\newblock {\em Canadian Journal of Mathematics}, 53, 08 2001.

\bibitem{DBLP:journals/moc/CurrieR11}
J.~Currie and N.~Rampersad.
\newblock A proof of {D}ejean's conjecture.
\newblock {\em Math. Comput.}, 80(274):1063--1070, 2011.

\bibitem{DBLP:journals/aam/CurrieR11}
J.~Currie and N.~Rampersad.
\newblock Recurrent words with constant {A}belian complexity.
\newblock {\em Adv. Appl. Math.}, 47(1):116--124, 2011.

\bibitem{CURRIE2012942}
J.~Currie and N.~Rampersad.
\newblock Fixed points avoiding {A}belian $k$-powers.
\newblock {\em J. Combin. Theory Ser. {A}}, 119(5):942--948, 2012.

\bibitem{DBLP:journals/ita/CurrieS09}
J.~Currie and K.~Saari.
\newblock Least periods of factors of infinite words.
\newblock {\em {RAIRO} Theor. Informatics Appl.}, 43(1):165--178, 2009.

\bibitem{DBLP:journals/tcs/CurrieV08}
J.~Currie and T.~Visentin.
\newblock Long binary patterns are abelian 2-avoidable.
\newblock {\em Theoret. Comput. Sci.}, 409(3):432--437, 2008.

\bibitem{DBLP:journals/ejc/DamanikL02}
D.~Damanik and D.~Lenz.
\newblock The index of {S}turmian sequences.
\newblock {\em European J. Combin.}, 23(1):23--29, 2002.

\bibitem{deLuca1996:sturmian_words_structure_combinatorics_arithmetics}
A.~{de Luca}.
\newblock {S}turmian words: structure, combinatorics, and their arithmetics.
\newblock {\em Theoret. Comput. Sci.}, 183(1):45--82, 1997.

\bibitem{DBLP:journals/jct/Dejean72}
F.~Dejean.
\newblock Sur un th{\'{e}}or{\`{e}}me de {T}hue.
\newblock {\em J. Comb. Theory, Ser. {A}}, 13(1):90--99, 1972.

\bibitem{Dekking1979181}
F.~Dekking.
\newblock Strongly non-repetitive sequences and progression-free sets.
\newblock {\em J. Combin. Theory Ser. {A}}, 27(2):181--185, 1979.

\bibitem{Didier99}
G.~Didier.
\newblock Caract{\'{e}}risation des {$N$}-{\'{e}}critures et application
  {\`{a}} l'{\'{e}}tude des suites de complexit{\'{e}} ultimement
  $n+c$\({}^{\mbox{\scriptsize\emph{ste}}}\).
\newblock {\em Theoret. Comp. Sci.}, 215(1-2):31--49, 1999.

\bibitem{DBLP:journals/ijfcs/DomaratzkiR12}
M.~Domaratzki and N.~Rampersad.
\newblock Abelian primitive words.
\newblock {\em Internat. J. Found. Comput. Sci.}, 23(5):1021--1034, 2012.

\bibitem{Primitive}
P.~D\"{o}m\"{o}si and M.~Ito.
\newblock {\em Context-Free Languages and Primitive Words}.
\newblock World Scientific, 2014.

\bibitem{DJP}
X.~Droubay, J.~Justin, and G.~Pirillo.
\newblock Episturmian words and some constructions by de {L}uca and {R}auzy.
\newblock {\em Theoret. Comput. Sci.}, 255, 2001.

\bibitem{durand_2003}
F.~Durand.
\newblock Corrigendum and addendum to ‘{L}inearly recurrent subshifts have a
  finite number of non-periodic factors’.
\newblock {\em Ergodic Theory and Dynamical Systems}, 23(2):663–669, 2003.

\bibitem{BEM79}
A.~E. Dwight Richard~Bean and G.~F. McNulty.
\newblock Avoidable patterns in strings of symbols.
\newblock {\em Pacific Journal of Mathematics}, 85(2):261--294, 1979.

\bibitem{EHRENFEUCHT1979101}
A.~Ehrenfeucht and D.~Silberger.
\newblock Periodicity and unbordered segments of words.
\newblock {\em Discrete Math.}, 26(2):101 -- 109, 1979.

\bibitem{EJS74}
R.~C. Entringer, D.~E. Jackson, and J.~A. Schatz.
\newblock On nonrepetitive sequences.
\newblock {\em J. Combin. Theory Ser. {A}}, 16(2):159--164, 1974.

\bibitem{Erdos}
P.~Erd\H{o}s.
\newblock Some unsolved problems.
\newblock {\em Magyar Tud. Akad. Mat., Kutat\'{o} Int. K\"{o}zl.}, 6:221--254,
  1961.

\bibitem{Evdokimov}
A.~Evdokimov.
\newblock Strongly asymmetric sequences generated by a finite number of
  symbols. ({R}ussian).
\newblock {\em Dokl. Akad. Nauk SSSR}, 179:1268--1271, 1968.

\bibitem{FM97}
S.~Ferenczi and C.~Mauduit.
\newblock Transcendence of numbers with a low complexity expansion.
\newblock {\em Journal of Number Theory}, 67:146--161, 1997.

\bibitem{tcs16}
G.~Fici, A.~Langiu, T.~Lecroq, A.~Lefebvre, F.~Mignosi, J.~Peltom\"{a}ki, and
  E.~Prieur-Gaston.
\newblock Abelian powers and repetitions in {Sturmian} words.
\newblock {\em Theoret. Comput. Sci.}, 635:16--34, 2016.

\bibitem{DBLP:journals/tcs/FiciMS17}
G.~Fici, F.~Mignosi, and J.~O. Shallit.
\newblock Abelian-square-rich words.
\newblock {\em Theoret. Comput. Sci.}, 684:29--42, 2017.

\bibitem{DBLP:journals/aam/FiciP019}
G.~Fici, M.~Postic, and M.~Silva.
\newblock Abelian antipowers in infinite words.
\newblock {\em Adv. Appl. Math.}, 108:67--78, 2019.

\bibitem{DBLP:journals/jct/FiciRSZ18}
G.~Fici, A.~Restivo, M.~Silva, and L.~Q. Zamboni.
\newblock Anti-powers in infinite words.
\newblock {\em J. Combin. Theory Ser. {A}}, 157:109--119, 2018.

\bibitem{FiSa14}
G.~Fici and A.~Saarela.
\newblock On the minimum number of abelian squares in a word.
\newblock {\em Combinatorics and Algorithmics of Strings, Dagstuhl Reports},
  4(3):34--35, 2014.

\bibitem{FW65}
N.~Fine and H.~Wilf.
\newblock Uniqueness theorem for periodic functions.
\newblock {\em Proc. Amer. Math. Soc.}, 16:109--114, 1965.

\bibitem{DBLP:journals/jct/FraenkelS98}
A.~S. Fraenkel and J.~Simpson.
\newblock How many squares can a string contain?
\newblock {\em J. Combin. Theory Ser. {A}}, 82(1):112--120, 1998.

\bibitem{pjm/1102784513}
J.~L. Gerver and L.~T. Ramsey.
\newblock {On certain sequences of lattice points.}
\newblock {\em Pacific J. Math.}, 83(2):357 -- 363, 1979.

\bibitem{DBLP:journals/ejc/GlenJWZ09}
A.~Glen, J.~Justin, S.~Widmer, and L.~Q. Zamboni.
\newblock Palindromic richness.
\newblock {\em European J. Combin.}, 30(2):510--531, 2009.

\bibitem{DBLP:journals/ijfcs/GocRRS14}
D.~Go\v{c}, N.~Rampersad, M.~Rigo, and P.~Salimov.
\newblock On the number of abelian bordered words (with an example of automatic
  theorem-proving).
\newblock {\em Internat. J. Found. Comput. Sci.}, 25(8):1097--1110, 2014.

\bibitem{DBLP:journals/jct/Graham73}
R.~L. Graham.
\newblock Covering the positive integers by disjoint sets of the form \{$[n
  \alpha + \beta]: n = 1, 2, ...$\}.
\newblock {\em J. Combin. Theory Ser. {A}}, 15(3):354--358, 1973.

\bibitem{HH00}
L.~Halbeisen and N.~Hungerb\"{u}hler.
\newblock An application of van der {W}aerden’s theorem in additive number
  theory.
\newblock {\em INTEGERS: Elect. Journ. Comb. Number Theory}, 0, paper A7, 2000.

\bibitem{HejdaSteinerZamboni15}
T.~Hejda, W.~Steiner, and L.~Q. Zamboni.
\newblock What is the {A}belianization of the {T}ribonacci shift?
\newblock Workshop on Automatic Sequences, Li{\` e}ge, May 2015.

\bibitem{DBLP:journals/eatcs/HenshallRS12}
D.~Henshall, N.~Rampersad, and J.~O. Shallit.
\newblock Shuffling and unshuffling.
\newblock {\em Bull. {EATCS}}, 107:131--142, 2012.

\bibitem{DBLP:journals/jct/Holub13}
{\v{S}}.~Holub.
\newblock Abelian powers in paper-folding words.
\newblock {\em J. Combin. Theory Ser. {A}}, 120(4):872--881, 2013.

\bibitem{DBLP:journals/dam/HolubS09}
{\v{S}}.~Holub and K.~Saari.
\newblock On highly palindromic words.
\newblock {\em Discrete Appl. Math.}, 157(5):953--959, 2009.

\bibitem{Hubert00}
P.~Hubert.
\newblock Suites {\'{e}}quilibr{\'{e}}es.
\newblock {\em Theoret. Comput. Sci.}, 242(1-2):91--108, 2000.

\bibitem{DBLP:journals/tcs/HuovaKS12}
M.~Huova, J.~Karhum{\"{a}}ki, and A.~Saarela.
\newblock Problems in between words and abelian words: $k$-abelian
  avoidability.
\newblock {\em Theoret. Comput. Sci.}, 454:172--177, 2012.

\bibitem{KaboreTapsoba07}
I.~Kabor{\'{e}} and T.~Tapsoba.
\newblock Combinatoire de mots r{\'{e}}currents de complexit{\'{e}} $n+2$.
\newblock {\em {RAIRO} Theor. Informatics Appl.}, 41(4):425--446, 2007.

\bibitem{DBLP:journals/ejc/KamaeH06}
T.~Kamae and H.~Rao.
\newblock Maximal pattern complexity of words over \emph{l} letters.
\newblock {\em European J. Combin.}, 27(1):125--137, 2006.

\bibitem{kamae_widmer_zamboni_2015}
T.~Kamae, S.~Widmer, and L.~Q. Zamboni.
\newblock Abelian maximal pattern complexity of words.
\newblock {\em Ergodic Theory and Dynamical Systems}, 35(1):142–151, 2015.

\bibitem{kamae_zamboni_2002}
T.~Kamae and L.~Zamboni.
\newblock Sequence entropy and the maximal pattern complexity of infinite
  words.
\newblock {\em Ergodic Theory Dynam. Systems}, 22(4):1191--1199, 2002.

\bibitem{DBLP:journals/tcs/KarhumakiPRW17}
J.~Karhum{\"{a}}ki, S.~Puzynina, M.~Rao, and M.~A. Whiteland.
\newblock On cardinalities of k-abelian equivalence classes.
\newblock {\em Theoret. Comput. Sci.}, 658:190--204, 2017.

\bibitem{DBLP:journals/ijfcs/KarhumakiPS13}
J.~Karhum{\"{a}}ki, S.~Puzynina, and A.~Saarela.
\newblock Fine and {W}ilf's theorem for $k$-abelian periods.
\newblock {\em Internat. J. Found. Comput. Sci.}, 24(7):1135--1152, 2013.

\bibitem{DBLP:conf/dlt/KarhumakiPW18}
J.~Karhum{\"{a}}ki, S.~Puzynina, and M.~A. Whiteland.
\newblock On abelian subshifts.
\newblock In M.~Hoshi and S.~Seki, editors, {\em {DLT} 2018}, volume 11088 of
  {\em Lecture Notes in Computer Science}, pages 453--464. Springer, 2018.

\bibitem{DBLP:journals/corr/abs-2012-14701}
J.~Karhum{\"{a}}ki, S.~Puzynina, and M.~A. Whiteland.
\newblock On abelian closures of infinite non-binary words.
\newblock {\em CoRR}, abs/2012.14701, 2020.

\bibitem{DBLP:journals/jct/KarhumakiSZ13}
J.~Karhum{\"{a}}ki, A.~Saarela, and L.~Q. Zamboni.
\newblock On a generalization of abelian equivalence and complexity of infinite
  words.
\newblock {\em J. Combin. Theory Ser. {A}}, 120(8):2189--2206, 2013.

\bibitem{DBLP:journals/actaC/KarhumakiSZ17}
J.~Karhum{\"{a}}ki, A.~Saarela, and L.~Q. Zamboni.
\newblock Variations of the {M}orse-{H}edlund theorem for \emph{k}-{A}belian
  equivalence.
\newblock {\em Acta Cybern.}, 23(1):175--189, 2017.

\bibitem{Ker92}
V.~Ker\"{a}nen.
\newblock Abelian squares are avoidable on 4 letters.
\newblock In {\em ICALP~1992}, volume 623 of {\em Lecture Notes in Comput.
  Sci.}, pages 41--52. Springer-Verlag, 1992.

\bibitem{keranen1}
V.~Ker\"{a}nen.
\newblock New abelian square-free {DT0L}-languages over 4 letters.
\newblock {\em Manuscript}, 2003.

\bibitem{DBLP:journals/tcs/Keranen09}
V.~Ker{\"{a}}nen.
\newblock A powerful abelian square-free substitution over 4 letters.
\newblock {\em Theoret. Comput. Sci.}, 410(38-40):3893--3900, 2009.

\bibitem{DBLP:journals/tcs/KociumakaRRW16}
T.~Kociumaka, J.~Radoszewski, W.~Rytter, and T.~Wale\'{n}.
\newblock Maximum number of distinct and nonequivalent nonstandard squares in a
  word.
\newblock {\em Theoret. Comput. Sci.}, 648:84--95, 2016.

\bibitem{DBLP:journals/tcs/KriegerS07}
D.~Krieger and J.~O. Shallit.
\newblock Every real number greater than 1 is a critical exponent.
\newblock {\em Theoret. Comput. Sci.}, 381(1-3):177--182, 2007.

\bibitem{DBLP:journals/ijac/LejeuneRR20}
M.~Lejeune, M.~Rigo, and M.~Rosenfeld.
\newblock The binomial equivalence classes of finite words.
\newblock {\em Int. J. Algebra Comput.}, 30(07):1375--1397, 2020.

\bibitem{DBLP:journals/aam/LejeuneRR20}
M.~Lejeune, M.~Rigo, and M.~Rosenfeld.
\newblock Templates for the \emph{k}-binomial complexity of the {T}ribonacci
  word.
\newblock {\em Adv. Appl. Math.}, 112, 2020.

\bibitem{DBLP:journals/aam/Li22}
S.~Li.
\newblock On the number of \emph{k}-powers in a finite word.
\newblock {\em Adv. Appl. Math.}, 139:102371, 2022.

\bibitem{Lietard}
F.~Lietard.
\newblock {\em \'Evitabilit\'e de puissances additives en combinatoire des
  mots}.
\newblock Phd thesis, Math\'ematiques [math], Universit\'e de Lorraine, 2020.

\bibitem{DBLP:conf/dlt/LietardR20}
F.~Lietard and M.~Rosenfeld.
\newblock Avoidability of additive cubes over alphabets of four numbers.
\newblock In N.~Jonoska and D.~Savchuk, editors, {\em {DLT} 2020}, volume 12086
  of {\em Lecture Notes in Computer Science}, pages 192--206. Springer, 2020.

\bibitem{LindMarcus95}
D.~Lind and B.~Marcus.
\newblock {\em An Introduction to Symbolic Dynamics and Coding}.
\newblock Camb. Univ. Press, New York, NY, USA, 1995.

\bibitem{Lot01}
M.~Lothaire.
\newblock {\em Combinatorics on Words}.
\newblock Cambridge Mathematical Library. Cambridge Univ. Press, 1997.

\bibitem{LothaireAlg}
M.~Lothaire.
\newblock {\em Algebraic Combinatorics on Words}.
\newblock Encyclopedia of Mathematics and its Applications. Cambridge Univ.
  Press, 2002.

\bibitem{Lot05}
M.~Lothaire.
\newblock {\em Applied Combinatorics on Words}.
\newblock Cambridge Univ. Press, 2005.

\bibitem{DBLP:journals/dm/MadillR13}
B.~Madill and N.~Rampersad.
\newblock The abelian complexity of the paperfolding word.
\newblock {\em Discr. Math.}, 313(7):831--838, 2013.

\bibitem{DBLP:conf/cwords/MasakovaP13}
Z.~Mas{\'{a}}kov{\'{a}} and E.~Pelantov{\'{a}}.
\newblock Enumerating {A}belian {R}eturns to {P}refixes of {S}turmian {W}ords.
\newblock In {\em {WORDS} 2013}, volume 8079 of {\em Lecture Notes in Computer
  Science}, pages 193--204. Springer, 2013.

\bibitem{Mi89}
F.~Mignosi.
\newblock Infinite words with linear subword complexity.
\newblock {\em Theoret. Comput. Sci.}, 65(2):221--242, 1989.

\bibitem{MignosiPirillo}
F.~Mignosi and G.~Pirillo.
\newblock Repetitions in the {F}ibonacci infinite word.
\newblock {\em RAIRO Theor. Inform. Appl.}, 26:199--204, 1992.

\bibitem{DBLP:journals/tcs/MignosiRS98}
F.~Mignosi, A.~Restivo, and S.~Salemi.
\newblock Periodicity and the golden ratio.
\newblock {\em Theor. Comput. Sci.}, 204(1-2):153--167, 1998.

\bibitem{DBLP:conf/icalp/MignosiS93}
F.~Mignosi and P.~S{\'{e}}{\'{e}}bold.
\newblock If a {D0L} language is $k$-power free then it is circular.
\newblock In A.~Lingas, R.~G. Karlsson, and S.~Carlsson, editors, {\em {ICALP}
  1993}, volume 700 of {\em Lecture Notes in Computer Science}, pages 507--518.
  Springer, 1993.

\bibitem{MoHe38}
M.~Morse and G.~A. Hedlund.
\newblock Symbolic dynamics.
\newblock {\em Amer. J. Math.}, 60:1--42, 1938.

\bibitem{Ochem2018}
P.~Ochem, M.~Rao, and M.~Rosenfeld.
\newblock Avoiding or limiting regularities in words.
\newblock In V.~Berth{\'e} and M.~Rigo, editors, {\em Sequences, Groups, and
  Number Theory}, pages 177--212. Springer International Publishing, 2018.

\bibitem{Pansiot}
J.-J. Pansiot.
\newblock Bornes inf\'{e}rieures sur la complexit\'{e} des facteurs des mots
  infinis engendr\'{e}s par morphismes it\'{e}r\'{e}s.
\newblock In M.~Fontet and K.~Mehlhorn, editors, {\em STACS 84}, pages
  230--240, Berlin, Heidelberg, 1984. Springer Berlin Heidelberg.

\bibitem{10.1145/321356.321364}
R.~J. Parikh.
\newblock On context-free languages.
\newblock {\em J. ACM}, 13(4):570–581, oct 1966.

\bibitem{DBLP:journals/combinatorics/ParreauRRV15}
A.~Parreau, M.~Rigo, E.~Rowland, and {\'{E}}.~Vandomme.
\newblock A new approach to the 2-regularity of the $l$-abelian complexity of
  2-automatic sequences.
\newblock {\em Electron. J. Comb.}, 22(1):1, 2015.

\bibitem{peltoPHD}
J.~Peltom\"{a}ki.
\newblock {\em Privileged Words and Sturmian Words}.
\newblock PhD thesis, University of Turku, Finland, 2016.

\bibitem{DBLP:journals/combinatorics/PeltomakiW17}
J.~Peltom{\"{a}}ki and M.~A. Whiteland.
\newblock A square root map on sturmian words.
\newblock {\em Electron. J. Comb.}, 24(1):P1.54, 2017.

\bibitem{DBLP:conf/mfcs/PeltomakiW20}
J.~Peltom{\"{a}}ki and M.~A. Whiteland.
\newblock All growth rates of abelian exponents are attained by infinite binary
  words.
\newblock In J.~Esparza and D.~Kr{\'{a}}l', editors, {\em {MFCS} 2020}, volume
  170 of {\em LIPIcs}, pages 79:1--79:10. Schloss Dagstuhl - Leibniz-Zentrum
  f{\"{u}}r Informatik, 2020.

\bibitem{DBLP:journals/aam/PeltomakiW20}
J.~Peltom{\"{a}}ki and M.~A. Whiteland.
\newblock Avoiding abelian powers cyclically.
\newblock {\em Adv. Appl. Math.}, 121:Article~102095, 2020.

\bibitem{PW20}
J.~Peltom{\"{a}}ki and M.~A. Whiteland.
\newblock On $k$-abelian equivalence and generalized {L}agrange spectra.
\newblock {\em Acta Arith.}, 194:135--154, 2020.

\bibitem{PELTOMAKI2020251}
J.~Peltomäki.
\newblock Abelian periods of factors of {S}turmian words.
\newblock {\em Journal of Number Theory}, 214:251 -- 285, 2020.

\bibitem{Pleasants}
P.~Pleasants.
\newblock Non-repetitive sequences.
\newblock {\em Proc. Cambridge Philos. Soc.}, 68:267--274, 1970.

\bibitem{DBLP:journals/combinatorics/Puzynina19}
S.~Puzynina.
\newblock Aperiodic two-dimensional words of small abelian complexity.
\newblock {\em Electron. J. Comb.}, 26(4):P4.15, 2019.

\bibitem{DBLP:journals/jcta/PuzyninaW22}
S.~Puzynina and M.~A. Whiteland.
\newblock Abelian closures of infinite binary words.
\newblock {\em J. Combin. Theory Ser. {A}}, 185:105524, 2022.

\bibitem{Pitheasfogg}
N.~Pytheas~Fogg.
\newblock {\em Substitutions in Dynamics, Arithmetics and Combinatorics},
  volume 1794 of {\em Lecture Notes in Math.}
\newblock Springer, 2002.

\bibitem{DBLP:journals/tcs/RampersadRS14}
N.~Rampersad, M.~Rigo, and P.~Salimov.
\newblock A note on abelian returns in rotation words.
\newblock {\em Theoret. Comput. Sci.}, 528:101--107, 2014.

\bibitem{RaSha_chapter}
N.~Rampersad and J.~Shallit.
\newblock Repetitions in words.
\newblock In V.~Berth{\'e} and M.~Rigo, editors, {\em Combinatorics, Words and
  Symbolic Dynamics}. Cambridge University Press, 2016.

\bibitem{DBLP:journals/tcs/Rao11}
M.~Rao.
\newblock Last cases of {D}ejean's conjecture.
\newblock {\em Theoret. Comput. Sci.}, 412(27):3010--3018, 2011.

\bibitem{DBLP:journals/tcs/Rao15}
M.~Rao.
\newblock On some generalizations of abelian power avoidability.
\newblock {\em Theoret. Comput. Sci.}, 601:39--46, 2015.

\bibitem{DBLP:journals/tcs/RaoRS15}
M.~Rao, M.~Rigo, and P.~Salimov.
\newblock Avoiding 2-binomial squares and cubes.
\newblock {\em Theoret. Comput. Sci.}, 572:83--91, 2015.

\bibitem{DBLP:journals/moc/RaoR16}
M.~Rao and M.~Rosenfeld.
\newblock Avoidability of long \emph{k}-abelian repetitions.
\newblock {\em Math. Comput.}, 85(302):3051--3060, 2016.

\bibitem{DBLP:journals/siamdm/RaoR18}
M.~Rao and M.~Rosenfeld.
\newblock Avoiding two consecutive blocks of same size and same sum over
  {\(\mathbb{z}\)}\({}^{\mbox{2}}\).
\newblock {\em {SIAM} J. Discret. Math.}, 32(4):2381--2397, 2018.

\bibitem{Ra82}
G.~Rauzy.
\newblock Suites {à} termes dans un alphabet fini.
\newblock {\em Séminaire de Théorie des Nombres de Bordeaux}, 25:1--16,
  1982-1983.

\bibitem{DBLP:journals/combinatorics/RichmondS09}
L.~B. Richmond and J.~O. Shallit.
\newblock Counting abelian squares.
\newblock {\em Electron. J. Comb.}, 16(1), 2009.

\bibitem{Richomme201179}
G.~Richomme, K.~Saari, and L.~Zamboni.
\newblock Abelian complexity of minimal subshifts.
\newblock {\em Journal of the London Mathematical Society}, 83(1):79--95, 2011.

\bibitem{DBLP:journals/aam/RichommeSZ10}
G.~Richomme, K.~Saari, and L.~Q. Zamboni.
\newblock Balance and abelian complexity of the {T}ribonacci word.
\newblock {\em Adv. Appl. Math.}, 45(2):212--231, 2010.

\bibitem{Rigo2014}
M.~Rigo.
\newblock {\em Formal Languages, Automata and Numeration Systems 1:
  Introduction to Combinatorics on Words}.
\newblock John Wiley \& Sons, 2014.

\bibitem{DBLP:journals/tcs/RigoS15}
M.~Rigo and P.~Salimov.
\newblock Another generalization of abelian equivalence: Binomial complexity of
  infinite words.
\newblock {\em Theoret. Comput. Sci.}, 601:47--57, 2015.

\bibitem{Rigo13}
M.~Rigo, P.~Salimov, and {\'E}.~Vandomme.
\newblock Some properties of abelian return words.
\newblock {\em J. Integer Seq.}, 16:13.2.5, 2013.

\bibitem{DBLP:conf/dlt/RigoSW22}
M.~Rigo, M.~Stipulanti, and M.~A. Whiteland.
\newblock Binomial complexities and {P}arikh-collinear morphisms.
\newblock In V.~Diekert and M.~V. Volkov, editors, {\em {DLT} 2022}, volume
  13257 of {\em Lecture Notes in Computer Science}, pages 251--262. Springer,
  2022.

\bibitem{DBLP:conf/mfcs/Rosenfeld16}
M.~Rosenfeld.
\newblock Every binary pattern of length greater than 14 is
  abelian-2-avoidable.
\newblock In {\em {MFCS} 2016}, volume~58 of {\em LIPIcs}, pages 81:1--81:11.
  Schloss Dagstuhl - Leibniz-Zentrum f{\"{u}}r Informatik, 2016.

\bibitem{DBLP:journals/jalc/Saarela09}
A.~Saarela.
\newblock Ultimately constant abelian complexity of infinite words.
\newblock {\em J. Autom. Lang. Comb.}, 14(3/4):255--258, 2009.

\bibitem{SAARI2010177}
K.~Saari.
\newblock Everywhere $\alpha$-repetitive sequences and {S}turmian words.
\newblock {\em European J. Combin.}, 31(1):177 -- 192, 2010.

\bibitem{DBLP:journals/ita/SamsonovS12}
A.~V. Samsonov and A.~M. Shur.
\newblock On abelian repetition threshold.
\newblock {\em {RAIRO} Theor. Informatics Appl.}, 46(1):147--163, 2012.

\bibitem{Shallit21}
J.~Shallit.
\newblock Abelian complexity and synchronization.
\newblock {\em Integers}, 21:A36, 2021.

\bibitem{shallit_2022}
J.~Shallit.
\newblock {\em The Logical Approach to Automatic Sequences: Exploring
  Combinatorics on Words with Walnut}.
\newblock London Mathematical Society Lecture Note Series. Cambridge University
  Press, 2022.

\bibitem{DBLP:journals/tcs/Simpson16}
J.~Simpson.
\newblock An abelian periodicity lemma.
\newblock {\em Theoret. Comput. Sci.}, 656:249--255, 2016.

\bibitem{simpson2018solved}
J.~Simpson.
\newblock Solved and unsolved problems about abelian squares.
\newblock {\em CoRR}, abs/1802.04481, 2018.

\bibitem{Turek15}
O.~Turek.
\newblock Abelian complexity function of the {T}ribonacci word.
\newblock {\em J. Integer Seq.}, 18:15.3.4, 2015.

\bibitem{DBLP:journals/tcs/Vandeth00}
D.~Vandeth.
\newblock Sturmian words and words with a critical exponent.
\newblock {\em Theoret. Comput. Sci.}, 242(1-2):283--300, 2000.

\bibitem{DBLP:journals/ejc/Vuillon01}
L.~Vuillon.
\newblock A characterization of {S}turmian words by return words.
\newblock {\em European J. Combin.}, 22(2):263--275, 2001.

\bibitem{DBLP:journals/jalc/Whiteland19}
M.~A. Whiteland.
\newblock Asymptotic abelian complexities of certain morphic binary words.
\newblock {\em J. Autom. Lang. Comb.}, 24(1):89--114, 2019.

\bibitem{whitelandPhD}
M.~A. Whiteland.
\newblock {\em On the $k$-Abelian Equivalence Relation of Finite Words}.
\newblock PhD thesis, University of Turku, TUCS Dissertations No 416, 2019.

\bibitem{ZamboniPersonal}
L.~Q. Zamboni.
\newblock Personal communication, 2018.

\bibitem{Zimin84}
A.~I. Zimin.
\newblock Blocking sets of terms.
\newblock {\em Sbornik: Mathematics}, 47(2):353--364, 1984.

\end{thebibliography}

\end{document}